\documentclass[english,leqno]{article}
\usepackage[letterpaper]{geometry}
\usepackage{amsmath,amssymb,amsthm}
\usepackage{graphicx}
\usepackage{nicefrac}
\usepackage{comment}
\usepackage{multirow}
\usepackage{adjustbox}
\usepackage{enumitem}
\usepackage[disable]{todonotes}
\usepackage{algorithm}
\usepackage[noend]{algpseudocode}
\usepackage{optidef}
\usepackage{wrapfig}
\usepackage{mdframed}
\usepackage[bottom]{footmisc} % Footnote at the bottom of page
\usepackage{thm-restate}
\usepackage[procnumbered,ruled,vlined,linesnumbered, algo2e]{algorithm2e}
\DontPrintSemicolon
\SetKw{KwAnd}{and}
\SetProcNameSty{textsc}
\SetFuncSty{textsc}
\usepackage{bm}
\let\oldnl\nl% Store \nl in \oldnl
\let\nl\oldnl% Remove line number for one line
\usepackage{hyperref}
\usepackage{comment}
\usepackage{xcolor}
\usepackage{xspace}
\usepackage{xfrac}
\usepackage[normalem]{ulem}

\usepackage{amsfonts}
\usepackage{fullpage}
\usepackage{epsfig}
\usepackage{xspace}
\usepackage{caption}
\usepackage{environ}
\usepackage{thmtools,thm-restate}
\usepackage{textcase}
\usepackage{xparse}

\definecolor{ForestGreen}{rgb}{0.1333,0.5451,0.1333}
\definecolor{DarkRed}{rgb}{0.8,0,0}
\definecolor{Red}{rgb}{1,0,0}

\declaretheorem[numberwithin=section]{theorem}
\declaretheorem[numberlike=theorem]{lemma}
\declaretheorem[numberlike=theorem]{corollary}

\declaretheorem[numberlike=theorem]{claim}
\declaretheorem[numberlike=theorem]{definition}

\usepackage[capitalize]{cleveref}
\crefname{claim}{Claim}{Claims}



\newcommand\set[1]{\{\,#1\,\}}

\newcommand{\XC}{\mathcal{C}}

\newcommand{\cost}[3]{\ensuremath{\mathsf{cost}_{#1}(#2#3)}}

\NewDocumentCommand{\COST}{ m m o }{%
	\ensuremath{%
		\IfValueTF{#3}{%
			\mathsf{cost}(#1#2#3)  % 3 arguments
		}{%
			\mathsf{cost}(#1#2)     % 2 arguments
		}%
	}%
}

\newcommand{\lp}[3]{\ensuremath{\mathsf{lp}_{#1}(#2#3)}}

\NewDocumentCommand{\LP}{ m m o }{%
	\ensuremath{%
		\IfValueTF{#3}{%
			\mathsf{lp}(#1#2#3)  % 3 arguments
		}{%
			\mathsf{lp}(#1#2)     % 2 arguments
		}%
	}%
}

\newcommand{\expected}[1]{\mathbb{E}\left[#1\right]}

\newcommand{\Cost}[1]{\ensuremath{\mathsf{cost}(#1)}}

\newcommand{\Lp}[1]{\ensuremath{\mathsf{lp}(#1)}}

\newcommand{\incl}[3]{f^{#1}(x_{{#2}{#3}})}

\newcommand{\REM}[1]{}

\newcommand{\cupdot}{\mathbin{\mathaccent\cdot\cup}}

\newcommand{\CC}{Correlation Clustering}
\newcommand{\CCC}{Constrained Correlation Clustering}
\newcommand{\charge}{\ensuremath{\kappa}}
\newcommand{\slackBad}{\ensuremath{\eta}}

\newcommand{\ALG}{\mathsf{ALG}}

\begin{document}
\title{\Large Constrained Correlation Clustering: Towards Optimality}

\author{
Sina Azizeddin \thanks{Sharif University of Technology, Iran, sina.azizeddin@gmail.com}
\and
Evangelos Kipouridis \thanks{Max Planck Institute for Informatics \&  Saarland Informatics Campus, Saarbrücken, Germany,
kipouridis@mpi-inf.mpg.de} 
\and
Nithin Varma \thanks{University of Cologne, Germany, nithvarma@gmail.com}} 
\date{}

\maketitle

\begin{abstract}
In the Correlation Clustering problem, we are given an undirected graph and are tasked with computing a clustering (partition of the nodes) that minimizes the number of violated pairs (edges across different clusters plus non-edges within clusters).
In the constrained version of this problem, the goal is to compute a clustering that satisfies additional hard constraints mandating certain pairs to be in the same cluster and certain pairs to be in different clusters.

In this work, we identify Constrained Correlation Clustering as a variant of Correlation Clustering for which optimal approximations might be within reach, and make progress towards this front.
Constrained Correlation Clustering is APX-Hard, and the optimal approximation factor is known to lie in $(\frac{24}{23},3]$.
We significantly tighten this gap, by showing that the optimal approximation factor lies in $[2,\frac{16}{7}-\gamma)$ for a small constant $\gamma>0$.
Our lower bound of $2$ shows a separation between Correlation Clustering (which admits an $1.485+\epsilon$ approximation) and Constrained Correlation Clustering\footnote{The same hardness result was obtained independently by Cao and Xu~\cite{cao2026clusterdeletionhardapproximate}.}.
Our upper bound of $\frac{16}{7}-\gamma$ uses the Sherali-Adams relaxation and goes beyond straightforward Triangle-Based analysis; more precisely, our algorithm belongs to a natural class of pivoting algorithms for which we prove that a straightforward Triangle-Based analysis cannot prove a better-than-$\frac{16}{7}$ approximation.

Finally, as a byproduct of our techniques, we completely resolve the approximability of Cluster Deletion.
Cluster Deletion is a well-studied special case of Constrained Correlation Clustering for which a $2$-approximation algorithm is known.
We show that this is optimal, as our lower bound holds even for this special case.

\end{abstract}

\thispagestyle{empty} %hide page number
\clearpage

\setcounter{page}{1} %set page number to 1
\section{Introduction}

%%Reworded INTRO

Clustering is a central problem in unsupervised learning, with many applications to machine learning and data mining. The high-level goal is to partition a set of elements into \emph{clusters} so that similar elements fall into the same cluster and dissimilar elements fall into different clusters. Among the various formalizations of clustering, one particularly elegant and natural formulation is that of \CC{}, proposed by Bansal, Blum, and Chawla~\cite{Bansal}.

\CC{} captures pairwise similarities and dissimilarities among a set $V$ of elements using a complete unweighted undirected labeled graph $(V, E^+ \uplus E^-)$, where the plus edges in $E^+$ represent similar pairs and the minus edges in $E^-$ represent dissimilar pairs. The objective is to determine a clustering that minimizes the total number of plus edges whose endpoints lie in different clusters and minus edges whose endpoints lie in the same cluster. \CC{} stands out among other clustering formulations because it does not require specifying the number of clusters as part of the input. Since its introduction, the problem has been widely studied and successfully applied to automated labeling~\cite{autoLabelAgrawal,autoLabelChakrabarti}, clustering ensembles~\cite{clusteringEnsembles}, community detection~\cite{commDet, communityVeldt}, disambiguation tasks~\cite{disambiguation}, duplicate detection~\cite{duplicate}, and image segmentation~\cite{image1,image2}, among others.

Real-world clustering tasks often come with prior knowledge that certain nodes must always be grouped together, while others must be kept apart. Several constrained clustering formulations have been proposed to model such requirements, including constrained $k$-means~\cite{cons1}, spectral clustering with constraints~\cite{Cons3}, constrained ranking and clustering~\cite{deterministicPivoting}, and constrained fuzzy clustering~\cite{cons4}. Constrained clustering is particularly valuable as it improves the quality of the resulting clusters, aligns outcomes with domain-specific knowledge, and reduces sensitivity to noisy data.

In this paper, we study \emph{Constrained Correlation Clustering}, a variant of \CC{} introduced by van Zuylen and Williamson~\cite{deterministicPivoting} that incorporates the notion of critical pairs of nodes. The input to \CCC{} consists of the labeled graph $(V, E^+ \uplus E^-)$ together with a collection of hard pairwise constraints: a \emph{must-link} constraint requires two nodes to be in the same cluster, while a \emph{cannot-link} constraint requires them to be in different clusters. Only clusterings that satisfy all hard constraints are considered valid, and the goal is to find an optimal valid clustering. Practical applications of \CCC{} include, for example, clustering news articles about the same event across different languages~\cite{CrosslingualGaelZ07}, where the hard constraints ensure that articles covering different events within the same language do not end up in the same cluster. There are also applications in DNA clone classification~\cite{dnaVertexDeletion}. In fact, in these cases there are no must-link constraints, and all minus edges are cannot-link constraints (also known as the \emph{Cluster Deletion} problem).

\subsection{Previous Results}
Bansal, Blum, and Chawla~\cite{Bansal}, who defined \CC{}, also proved that it is NP-Hard. They further provided a deterministic constant-factor approximation, where the constant is larger than 15,000. The approximation guarantee was subsequently improved further, by relying on LP-rounding techniques. Charikar, Guruswami and Wirth~\cite{cutRadius} gave a deterministic $4$-approximation, while also showing that the problem is APX-Hard. Ailon, Charikar and Newman~\cite{pivoting} presented a randomized $2.5$-approximation, and Chawla, Makarychev, Schramm and Yaroslavtsev~\cite{NearOptimal2} gave a deterministic $2.06$-approximation. The last result is nearly optimal among algorithms that round the natural LP formulation, since its integrality gap is at least $2$. A breakthrough result by Cohen-Addad, Lee and Newman~\cite{sub2} used the Sherali-Adams relaxation to obtain a $(1.994+\epsilon)$-approximation and broke the $2$ barrier. The approximation guarantee was later improved to $1.73+\epsilon$ by Cohen-Addad, Lee, Li and Newman \cite{apx173}, and even to $1.485+\epsilon$ by~\cite{clusterLP}. 
The current-state of the art is an adaptation of the $1.485+\epsilon$ approximation that runs in linear~\cite{sublinearMikkelSTOC}, and even sublinear~\cite{sublinearMikkelSTOC} time. 
There is also a combinatorial $1.847$-approximation by~\cite{combinatorialMikkel} for the problem.

Correlation Clustering has also been studied in different settings such as dynamic algorithms~\cite{dynNik}, sublinear and streaming algorithms~\cite{sublinear, dynamicStream, BehMany, BehSingle, MakarySingle, dynamicStream}, parameterized algorithms~\cite{fpt}, massively parallel computation~(MPC) algorithms~\cite{parallel, sub3parallel}, and differentially private algorithms~\cite{differentialPrivacy}.

The approximation of the Cluster Deletion problem (special case of \CCC{} where there are no must-link constraints, and all minus edges are cannot-link constraints) has been extensively studied \cite{cutRadius, deterministicPivoting, dnaVertexDeletion, local, communityVeldt, deterministic, veldtFaster2, clusterDeletionSpecialClassesBoundedClique}.
The current best is a $2$ approximation, by Veldt, Gleich, and Wirth \cite{communityVeldt}.
Subsequent research focused on faster algorithms even at the expense of worse approximation factors~\cite{deterministic, veldtFaster2}, and better-than-$2$ approximations for particular graph classes~\cite{clusterDeletionSpecialClassesBoundedClique}. 

For \CCC{} the state-of-the-art is a deterministic $3$-approximation, presented by van Zuylen and Williamson~\cite{deterministicPivoting}, who introduced the problem.
Due to the significance of the problem, a follow-up work by Fischer, Klausen, Kipouridis and Thorup~\cite{EvangelosSTACS} focused on faster algorithms for \CCC{}, even at the expense of the approximation guarantee ($16$); this was improved by Veldt to a $3+\varepsilon$ approximation \cite{VeldtConstrained}.
Finally, in \cite{kalavas} the authors show how an adaptation of the Linear Program (Cluster LP) solved in \cite{clusterLP} would imply a better-than-$2$ approximation for \CCC{}. 

As a final remark, we point out that it is not possible, in general, to handle the must-link constraints by simply merging the nodes in connected components induced by these constraints into single nodes, as this results in a weighted instance outside the problem.

\subsection{Our Contribution and Techniques}
The optimal approximation factor for Constrained Correlation Clustering is known to lie in $(\frac{24}{23},3]$ \cite{clusterLP, deterministicPivoting}.
In this work, we tighten the gap to $[2,\frac{16}{7}-\gamma)$ for some constant $\gamma>0$, hinting that $2$ might be the optimal approximation for \CCC{}\footnote{On a more informal note, for both Correlation Clustering and Chromatic Correlation Clustering the only known lower bound is $\nicefrac{24}{23}$~\cite{clusterLP}, while they both admit better-than-$2$ approximations~\cite{clusterLP,chromaticClusterLP}, making it difficult to conjecture the actual optimal approximation factor. To the best of our knowledge, all other variants of Correlation Clustering (such as Pseudometric-Weighted Correlation Clustering~\cite{barrierPseudometric}) also lack ``natural'' lower bounds.}.

\paragraph{Lower Bound}
Our lower bound of $2$ shows a separation between Correlation Clustering (which admits an $1.485+\epsilon$ approximation \cite{clusterLP}) and Constrained Correlation Clustering. It is conditioned on the Unique Games~\cite{ugc}.
Furthermore, combined with the work of \cite{kalavas}, our work implies that the Cluster LP (the most important tool in the state of the art approximation for Correlation Clustering) cannot be adapted in the constrained setting.

\begin{restatable}{theorem}{hardness} \label{thm:hardness}
    Under the Unique Games Conjecture, there is no $(2 - \varepsilon)$-factor approximation algorithm for \CCC{}, for any constant $\varepsilon>0$.
    Furthermore, under the Unique Games Conjecture, there is no $(2 - \varepsilon)$-factor approximation algorithm for Cluster Deletion, for any constant $\varepsilon>0$.
\end{restatable}

It is known that Cluster Deletion is approximable within a factor $2$~\cite{communityVeldt}.
Our lower bound immediately implies that this is optimal.

\begin{corollary}
There exists a polynomial time algorithm for Cluster Deletion whose approximation factor is (the optimal) $2$.
\end{corollary}

We note that the same lower bound result was obtained independently by Cao and Xu~\cite{cao2026clusterdeletionhardapproximate}.

\paragraph{The barrier of $\frac{16}{7}$}
A particular type of algorithms, called pivoting algorithms, have been extensively used for Correlation Clustering~\cite{pivoting, NearOptimal2, sub2, apx173, clusterLP}.
A common type of analysis for such algorithms is the Triangle-Based analysis. From a high-level view, in this analysis we bound the ratio of the cost of the algorithm to the budget given by a Linear Program (LP, \Cref{fig:LP}), independently for every triplet of nodes (triangle).
The overall approximation is upper bounded by the maximum such ratio.

To address hard constraints, a natural class of pivoting algorithms is presented in \Cref{alg:strictPivot}. We call these algorithms \emph{strict pivoting} algorithms. Let us use the term \emph{supernode} to refer to a maximal set of nodes connected with must-link constraints. From a high-level view:
\begin{itemize}
    \item The LP relaxation in \Cref{fig:LP} gives a distance in $[0,1]$ for every pair of nodes, such that must-link pairs have distance $0$ and cannot-link pairs have distance $1$.
    \item At every step of the algorithm, until all nodes are clustered, we select a supernode $p$ (pivot) at random, with probability proportional to the number of nodes it contains. 
    \item For every unclustered supernode $s$, we decide whether to include it in the cluster of $p$ with probability related to the LP distance of $p$ and $s$, and the type of edges connecting the supernode of $p$ and $s$.
    Crucially, if the LP distance is $0.5$ or larger, we do not cluster $s$ together with $p$ (hence the term \emph{strict pivoting}).
\end{itemize}
Notice that hard constraints are always satisfied.
This follows trivially for must-link constraints.
For cannot-link constraints, notice that within a cluster every node is at distance less than $0.5$ from the pivot.
Therefore, by triangle inequality, the diameter of any cluster is less than $1$; as cannot-link pairs have distance exactly $1$, this means that we do not violate cannot-link constraints.

We show that the straightforward Triangle-Based analysis used for such algorithms cannot give a better-than-$\frac{16}{7}$ approximation, while also presenting an algorithm that achieves $\frac{16}{7}$ approximation.
It is worth noting that for Correlation Clustering, a similar barrier exists ($2.025$), but no algorithm within the corresponding framework achieves this bound (best known approximation is $\approx 2.06$); see~\cite{NearOptimal2}.
Similarly, for Chromatic Correlation Clustering, within the corresponding framework the lower bound $2.11$ and upper bound $2.15$ do not match~\cite{barrierPseudometric}. We are only aware of Pseudometric-Weighted Correlation Clustering where within its corresponding framework an optimal $\nicefrac{10}{3}$ approximation has been designed (however, in this case, no better algorithm is known, while for \CCC{} we show a $\frac{16}{7}-\gamma$ approximation).

\paragraph{Approximation beyond $\frac{16}{7}$}
Our upper bound of $\frac{16}{7}-\gamma$ (for a small $\gamma>0$) goes beyond the standard Triangle-Based analysis.
We note however that the algorithm itself uses the framework of \Cref{alg:strictPivot}.
Unlike our analysis for the $\frac{16}{7}$-approximation algorithm, for which it is sufficient to use the properties of the standard LP (\Cref{fig:LP}), our $\frac{16}{7}-\gamma$ approximation requires the usage of the Sherali-Adams relaxation (\Cref{fig:SA}).

The argument uses a global charging scheme, inspired from the first better-than-$2$ approximation for Correlation Clustering~\cite{sub2}\footnote{
For the knowledgeable reader, the core argument in ~\cite{sub2} shows that we cannot have too many ``bad'' triangles in any subgraph of $4$ nodes; in our case the argument gets more technical and requires subgraphs of $6$ nodes.
}

\begin{restatable}{theorem}{mainTheorem}
\label{thm:mainTheorem}
There exists a $\frac{16}{7}-\gamma$ approximation for \CCC{}, for some constant $\gamma>0$.
\end{restatable}

\paragraph{Comparison with Correlation Clustering}
There is a big literature on (unconstrained) Correlation Clustering.
It is therefore imperative to discuss if/how techniques that worked in the unconstrained setting can be expected to provide results for the constrained setting as well.
We focus on techniques that led to better-than-$3$ approximations, as a $3$-approximation for \CCC{} was already known before our work.

\begin{itemize}
    \item \textit{Cluster LP}: The Cluster LP is an exponentially large LP that has been shown to be solvable in polynomial time~\cite{clusterLP}. It has been used to achieve the state of the art $1.485+\epsilon$ approximation for Correlation Clustering.
    However, as shown in \cite{kalavas}, solving this LP in the constrained setting would imply a better-than-$2$ approximation for \CCC{}.
    As per our \Cref{thm:hardness}, this would contradict the Unique Games Conjecture.
    \item \textit{Correlated Rounding}: The first algorithm achieving a $2$-approximation for Correlation Clustering (which was improved to $1.994$ in the same paper~\cite{sub2}) used a pivoting algorithm and correlated rounding. Specifically, for any two nodes $u,v$ that differ from the pivot $p$, the decision on whether to include $u$ in the cluster containing $p$ was not independent from the decision regarding $v$. Unfortunately, our $\frac{16}{7}$ barrier holds even if correlated rounding is used. Furthermore, implementing correlated rounding that would respect the hard constraints seems to be the main barrier in solving the Constrained Cluster LP, and is thus unlikely.
    \item \textit{Set-based Rounding}: An important ingredient to improve the approximation from $1.994$ to $1.73$ was the set-based clustering~\cite{apx173}. According to the authors, this uses the ``the power of [...] the correlated rounding technique'', which we already highlighted as a challenge in the constrained setting.
    \item \textit{Local Search}: The local search technique from~\cite{combinatorialMikkel} assumes access to some ``near-optimal'' clusters and provides a better-than-$2$ approximation.
    Even though obtaining such clusters is possible for Correlation Clustering, in the constrained setting it would automatically imply a better-than-$2$ approximation~(\cite{kalavas}), which would contradict the Unique Games Conjecture~\cite{ugc}.
    \item \textit{Specialized Rounding}: Pairing the pivoting algorithm with specialized rounding functions has been used extensively~\cite{pivoting,NearOptimal2,sub2,apx173,clusterLP} for \CC{}.
    However, for \CCC{}, only a very simple rounding (rounding to the closest integer) has been employed~\cite{deterministicPivoting} so far.
    A possible reason is that the must-link constraints enforce (a special case of) correlated rounding, which turned out to be a significant challenge for our work as well.
    Concretely, all previous approaches use two different rounding functions, depending on whether a plus or a minus edge is rounded. In order to satisfy must-link constraints however, we need to round sets of multiple edges simultaneously, while it is not necessary for all of them to be plus (resp.\ minus).
    
    We show that using a convex combination of two particular functions (one for plus and one for minus edges), with the coefficients depending on the number of plus (resp. minus) edges in the set of edges we round, it suffices to consider cases where all edges are plus or all are minus\footnote{On a technical note which may only be clear after reading \Cref{sec:barrier}: this is again non-trivial, exactly because of the must-link constraints, which are the sole reason a non-linear (in the coefficients of the convex combination) behavior is introduced.}.
    
    Perhaps surprisingly, the key property of the particular functions we use, that allows for the aforementioned argument, is that if a plus and a minus edge have the same LP distance, then we are more likely to cluster together the endpoints of the minus edge, rather than the endpoints of the plus edge.
    This was not the case, for example, in the work of~\cite{NearOptimal2} that introduced the specialized rounding framework.
\end{itemize}

\paragraph{Summing up}
To summarize the contributions of our work, we identify \CCC{} as a variant of \CC{} for which an optimal approximation might be within reach.
We conjecture the optimal approximation to be $2$, and prove that it lies in $[2,\frac{16}{7})$.
At the same time, our results demonstrate that the Cluster LP cannot be $(1+\varepsilon)$ approximated in the constrained setting, and is thus unlikely to help in designing an optimal approximation.
Furthermore, the optimality cannot be achieved using specialized rounding and a straightforward Triangle-Based analysis, even if the actual optimal is larger than $2$; that is because we show that, unlike our upper bound, this framework cannot prove better-than-$\frac{16}{7}$ approximations.
Finally, as a byproduct of our work, we completely settle the approximability of the Cluster Deletion problem, showing that the optimal approximation factor is $2$.

\subsection{Related Work}
Closely related to the problems considered by us are hierarchical clustering problems under similar constrained settings.
In particular, for each pair $\{u,v\}$ we are given an upper and a lower bound regarding the distance of $u$ and $v$ in the output (Farach, Kannan, and Warnow~\cite{robust}).
Ailon and Charikar, in~\cite{charikarTree}, make progress on constrained hierarchical clustering, but optimal algorithms are still elusive.
We note that \CCC{} is a special case of constrained hierarchical clustering (where the hierarchy is trivial).
Tree reconstruction problems with lower bound constraints (similar to the cannot-link constraints) have also been studied \cite{Ma3Infinity, l0TreeLower}.

%Ranking problems have also been studied under constrained settings similar to ours (see~\cite{deterministicPivoting} and references therein). In particular, their hard constraints are pairs that must have a particular ordering in the output.

\subsection{Organization}
The rest of the paper has the following structure.
In \Cref{sec:prelims}, we introduce some necessary definitions and results.
In \Cref{sec:hardness}, we present our lower bound of $2$.
In \Cref{sec:barrier}, we present our algorithmic framework and the $\frac{16}{7}$ barrier.
In \Cref{sec:analysis}, we present our algorithm and prove that it achieves a $\frac{16}{7}$ approximation, while in \Cref{sec:beyond}, we prove a $\frac{16}{7}-\gamma$ approximation, for a small constant $\gamma>0$.

\section{Preliminaries}\label{sec:prelims}
The graphs $(V,E^+\uplus E^-)$ we consider in this paper are complete, unweighted, undirected, and every edge is labeled plus or minus (it belongs to one of $E^+$ or $E^-$, respectively).
Throughout the paper we refer to a set $F \subseteq \binom{V}{2}$ as the friendly (must-link) constraints, and to a set $H \subseteq \binom{V}{2}$ as the hostile (cannot-link) constraints. Together, they constitute the so-called \emph{hard constraints}. 

We typically set~$n = |V|$.
Let $\binom{A}{2}$ denote the set of all size-$2$ subsets of a set $A$.
We often abbreviate the (unordered) set $\set{u, v}$ by $uv$, and the (unordered) set $\set{u, v, w}$ by $uvw$.
We use $\cupdot$ to denote unions of disjoint sets.

A clustering $\mathcal{C}=\set{C_1,\ldots, C_k}$ is a partition of $V$; each set $C_i$ is called a cluster.
We denote by $\mathcal{C}(u)$ the cluster of $\mathcal{C}$ containing $u$.

\begin{definition}[Constrained Correlation Clustering]
Given $(V,E^+\uplus E^-,F,H)$, where $E^+, E^-, F, H \subseteq \binom{V}{2}$, compute a partition of $V$ (clustering) $\mathcal{C}=\set{C_1,\ldots, C_k}$ such that:
\begin{itemize}
    \item if $uv \in F$, then both $u$ and $v$ are in the same cluster $C_i$;
    \item if $uv \in H$, then $u$ and $v$ are in different clusters $C_i\ni u, C_j\ni v$, where $i \neq j$;
    \item $\XC$ minimizes $|\set{uv\in E^+ \mid \XC(u) \ne \XC(v)}| + |\set{uv\in E^- \mid \XC(u) = \XC(v)}|$ out of all possible clusterings that satisfy the above two conditions.
\end{itemize}
\end{definition}

As per~\cite{deterministicPivoting}, who introduced \CCC{}, we assume that the constraints are consistent:
\begin{itemize}
    \item if $ab,bc\in F$, then also $ac\in F$ ($F$ is transitive);
    \item if $ab\in F$ and $bc\in H$ then $ac\in H$;
    \item $F\cap H = \emptyset$;
    \item if $ab\in F$ then $ab\in E^+$; similarly if $ab\in H$ then $ab\in E^-$.
\end{itemize}
It is easy to see that the set of feasible clusterings does not change if we minimally extend $F,H$ to satisfy the first two conditions.
It is also straightforward that given the first two conditions, the third condition holds if and only if the instance is satisfiable (there exists a feasible clustering).
Finally, if an algorithm $\alpha$ approximates satisfiable instances that satisfy the fourth condition, then it $\alpha$ approximates any satisfiable instance.

From now on, we use the term \emph{supernode} to denote a maximal set of nodes that are pairwise friendly.
For a node $u\in V$, we use $s(u)$ to denote the supernode containing $u$, and $s(V)$ to denote the set of all supernodes.

\paragraph{The standard LP:}
In the standard LP (\Cref{fig:LP}), we have a variable $x_{uv}$ for every pair $uv$ which can be interpreted as the distance between $u,v$.
In an integral clustering, if $u,v$ are in the same cluster then $x_{uv}=0$, and if they are in different clusters then $x_{uv}=1$.
The objective function of this LP (\Cref{fig:LP}) follows directly by the definition of \CCC{}.
We enforce triangle inequality (which means that if $uw$ are in the same cluster, and $vw$ are also in the same cluster, then $uv$ are in the same cluster) and hardcode the values $0$ and $1$ in the case of hard constraints.

\begin{figure}
    \centering
    \begin{align*}
\min \quad & \sum_{(i, j) \in E^+} x_{uv} + \sum_{(i, j) \in E^-} (1-x_{uv})  \\
\mbox{s.t. } & x_{uv} \leq x_{uw} + x_{vw} && \forall u,v,w \in V \\
& x_{uv} = 0
&& \forall uv \in F \\
& x_{uv} = 1
&& \forall uv \in H\\
& x_{uv} \in [0, 1]
&& \forall u,v \in V
\end{align*}
    \caption{The standard LP relaxation for \CCC{}}
    \label{fig:LP}
\end{figure}

\paragraph{Sherali-Adams relaxation:}
Similarly to the standard LP, we have the Sherali-Adams relaxation.
Here it is convenient to introduce the variables $y_{uv} = 1-x_{uv}$.
We note that any solution to the Sherali-Adams relaxation is also a feasible solution to the standard LP; we highlight this by adding the triangle-inequality on the $x$ variables, even though it is implied by the rest of the constraints.
Furthermore, the optimal cost of the Sherali-Adams relaxation is also a lower bound for the optimal cost of the \CCC{} instance (although maybe higher than the optimal cost of the standard LP).

More formally, the definition of the $r$-rounds of the Sherali-Adams relaxation is in \Cref{fig:SA}.
In our paper, we use $r=6$.
For disjoint sets $S_1, \dots,
S_{\ell} \subseteq V$ such that $\sum_{i=1}^{\ell} |S_i| \leq r=6$, we
have a variable $y_{S_1 | S_2 | \dots | S_{\ell}}$ indicating whether the optimal partition induced by $S_1 \cup \dots \cup
S_{\ell}$ is exactly $(S_1, \dots, S_{\ell})$. For example, for two vertices
$u$ and $v$, $y_{u|v}$ indicates whether
$u$ and $v$ are in different clusters in the optimal solution and
$y_{uv}$ indicates whether they are in the same cluster,
so that $y_{u|v} + y_{uv} = 1$.
The main benefit of the Sherali-Adams relaxation is that it allows us to naturally argue about the relation of more than $2$ nodes (but only up to $r=6$ nodes).
For example, we can use arguments of the form $y_{uv} = y_{uv|w} + y_{uvw}$.

\begin{figure}
\begin{align*}
  \min \quad & \sum_{ij \in E^+} x_{ij}  + \sum_{ij \in E^-} (1-x_{ij}) \\
\mbox{s.t.} \quad & y_{T_1|\ldots|T_k}   = \sum_{\substack{S_1,\ldots, S_{\ell}: \\ S = S_1 \cupdot \ldots \cupdot S_{\ell} \\ \text{and } T_i = S_i \cap T~ \forall i \in [k]}} y_{S_1|S_2|\ldots|S_{\ell}} \quad && \forall 
T \subseteq  S \subseteq V, |S| \leq r, \text{ and } T= T_1 \cupdot \ldots \cupdot T_k\\
  &y_{\emptyset} = 1  \\
  & x_{uv} \leq x_{uw} + x_{vw} && \forall u,v,w \in V \\
& x_{uv} = 0
&& \forall uv \in F \\
& x_{uv} = 1
&& \forall uv \in H\\
  &y  \geq 0 
\end{align*}
    \caption{The $r$-rounds of the Sherali-Adams relaxation. The last constraint requires the $y$-variables across all possible subscripts to be nonnegative. In this paper we use $r=6$.} \label{fig:SA}
\end{figure}

\paragraph{Inclusion probabilities:}
We define two functions $f^+, f^-: [0, 1] \to [0, 1]$ that we use throughout the paper.
They determine the inclusion probabilities in our pivot based algorithms.
$$
f^+(x) = \begin{cases}
	\sqrt{1 - \frac{20}{9} x} & \text{if } x < 0.25 \\
	\frac{13}{12} - \frac{5}{3} x & \text{if } 0.25 \leq x < 0.5,
	\qquad
	f^-(x) = \begin{cases}
		1 - x & \text{if } x < 0.5 \\
		0 & \text{if } x \geq 0.5
	\end{cases} \\
	0 & \text{if } x \geq 0.5
\end{cases}
$$

We now state several basic analytic properties of the functions $f^+,f^-$ that are used in the approximation analysis (their proofs are found in \Cref{app:proofsForFunctions}).

\begin{restatable}{lemma}{derCont}\label{lemma:der-cont}
	The function $f^+(x)$ is continuous on $[0,0.5)$ and has a continuous first derivative on this domain.
\end{restatable}

\begin{restatable}{lemma}{funcIneq}\label{lemma:func-ineq}
	For all $x$:
	\begin{itemize}
		\item[$(i)$] If $x\in[0,0.5)$, then $\sqrt{1-\frac{20}{9} x} \leq f^+(x) \leq \frac{13}{12} - \frac{5}{3}x$.
		\item[$(ii)$] If $x\in[0,0.5)$, then $1 - \frac{10}{9}x \ge f^+(x) \ge 1 - \frac{3}{2}x$.
	\end{itemize}
\end{restatable}

\begin{restatable}{lemma}{fProduct}
	\label{lemma:f-product}
	\REM{For all $x, y \in [0, 0.5)$, we have
	$$
	f(x) \cdot f(y) \leq f(\frac{x+y}{2})^2.
	$$}
    The function $\ln (f^+(x))$ is a concave function on $[0, 0.5)$.
\end{restatable}

\begin{restatable}{lemma}{funcInequality}\label{lemma:funcsInequality}
    For all $x\in [0, 1]$, we have $f^+(x) \leq f^-(x)$.
\end{restatable}

\section{Hardness of Approximation} \label{sec:hardness}

In this section, we show that it is UGC-hard to obtain a better-than-$2$-approximation algorithm for the Constrained Correlation Clustering problem (where UGC is the Unique Games Conjecture).

\hardness*

\begin{proof}
We show that a $(2 - \varepsilon)$ approximation for \CCC{} would imply a $(2 - \frac{\varepsilon}{2})$ approximation for Min Vertex Cover, which is impossible assuming the Unique Games Conjecture (UGC)~\cite{vc-hard-to-approx}.

\paragraph{Reduction:}
Given an instance $G = (V,E)$ of Min Vertex Cover, with $n=|V|$, we transform it into an instance $\mathcal{C}_G = (V',E^+\uplus E^-, F, H)$ of Constrained Correlation Clustering as follows. 
The instance $\mathcal{C}_G$ is defined on a set of nodes $V'=V \uplus V''$, where $|V''| = n^3$. The hostile constraints and the non-edges in $\mathcal{C}_G$ are the edges in $G$, that is $H = E^- = E$.
It follows that the plus edges $E^+$ are the nonedges of $G$ and all the pairs with at least one endpoint in $V''$.
Finally, there are no friendly constraints, that is $F = \emptyset$.

Notice that the instance $\mathcal{C}_G$ is satisfiable, e.g. by making every node a singleton.
Furthermore, as $F=\emptyset$ and $H=E^-$, $\mathcal{C}_G$ is equivalent to the instance $(V',(E^+,E^-))$ of Cluster Deletion. This implies that the lower bound we derive for \CCC{} immediately extends to Cluster Deletion.

Now assume we have a $(2 - \varepsilon)$-approximation algorithm $A$ for Constrained Correlation Clustering. We use this algorithm as a subroutine to obtain a $(2-\varepsilon/2)$-approximation algorithm for Min Vertex Cover.
In particular, given an input graph $G$ for Min Vertex Cover, we first run $A$ on $\mathcal{C}_G$ and obtain clustering $\mathcal{C}$. We return every node in $V$ that is not contained in a cluster of $\mathcal{C}$ with size at least $0.99n^3$ (clearly there can be at most one such large cluster).

Note that the nodes in $V$ that are not in the output are all contained in the same cluster (of size at least $0.99n^3$).
Therefore, by definition of $H$, they form an independent set in graph $G$, meaning that our output is a valid vertex cover.

We now need to prove that the size of our output is at most $(2-\varepsilon/2)$ times the size of the Min Vertex Cover.
We first need the following claims:

\begin{claim} \label{clm:optVertexCover}
The cost of the optimal clustering in $\mathcal{C}_G$ is at most $v \cdot n^3 + \binom{n}{2}$, where $v$ is the size of the minimum vertex cover in $G$.
\end{claim}
\begin{proof}
Consider a maximum independent set $I \subseteq V$ in $G$. Note that $V \setminus I$ is a minimum vertex cover in $G$. Consider a clustering, where $I \cup V''$ is a cluster and each vertex in $V \setminus I$ forms its own singleton cluster. This clustering is valid and has cost at most $|V \setminus I| \cdot |V''| + \binom{n}{2} = v \cdot n^3 + \binom{n}{2}$.
\end{proof}

\begin{claim}
There exists a cluster $C_{big} \in \mathcal{C}$ such that $|C_{big} \cap V''| \ge 0.99|V''|$.
\end{claim}
\begin{proof}
It suffices to prove that there cannot exist a subset $\mathcal{S} \subseteq \mathcal{C}$ such that $$0.01|V''| < |(\bigcup_{C \in \mathcal{S}} C) \cap V''| < 0.99|V''|.$$
Indeed, if this were the case, the cost of $\mathcal{C}$ would be at least $(0.01)^2 |V''|^2 = \Omega(n^6)$ due to the plus edges with one endpoint in $V'' \cap (\bigcup_{C \in \mathcal{S}} C)$ and the other endpoint in $V''\setminus (\bigcup_{C \in \mathcal{S}} C)$. This is a factor $\Omega(n^2)$ worse off than the optimal clustering with cost $O(n^4)$ (see~\cref{clm:optVertexCover}). For large enough $n$, this contradicts the assumption that $A$ is a $(2-\varepsilon)$-approximation algorithm for Constrained Correlation Clustering.
\end{proof}

To analyze $|V\setminus C_{big}|$, we first modify $\mathcal{C}$ and obtain $\mathcal{C'}$ as follows: For each cluster $C \neq C_{big}$ satisfying $C \cap V'' \neq \emptyset$, we remove the vertices in $C \cap V''$, meaning that we are left with $C\setminus V''$ instead of $C$.
Furthermore, we replace $C_{big}$ with $C_{big}\cup V''$.
The modified clustering $\mathcal{C'}$ is valid as the only vertices we move across clusters are contained in $V''$, for which there are no hard constraints. 
Now notice that the only edges with possibly different cost contributions in $\mathcal{C}$ and in $\mathcal{C'}$ are edges with one endpoint in $V''\setminus C_{big}$.
For each node $u\in V''\setminus C_{big}$, its cost contribution was at least $|C_{big}|\ge 0.99n^3$, while its new cost contribution is at most $n$.
Thus $\mathcal{C'}$ is at least as good as $\mathcal{C}$. In particular, it is also a $(2-\varepsilon)$-approximate clustering.

The cost of $\mathcal{C'}$ is at least $|V\setminus C_{big}| \cdot n^3$.
By \Cref{clm:optVertexCover} the optimal cost is at most $v\cdot n^3 + \binom{n}{2}$, where $v$ is the size of the minimum vertex cover in $G$.
Therefore $|V\setminus C_{big}| \cdot n^3 \leq (2 - \varepsilon)\cdot (v\cdot n^3 + \binom{n}{2}),$ implying that $|V\setminus C_{big}| \leq (2 - \frac{\varepsilon}{2})\cdot v$, for large enough $n$. This violates the Unique Games Conjecture.
\end{proof}

\section{\texorpdfstring{Algorithmic framework and $16/7$ barrier}{Algorithmic framework and 16/7 barrier}} \label{sec:barrier}

In this section, we present the algorithmic framework we use throughout this paper (\Cref{alg:strictPivot}).
We present it in full generality, using a placeholder function $p(\cdot, \cdot, \cdot)$; this function may differ across different algorithms in the framework.
We prove that for any algorithm in this framework, a particular type of analysis (standard Triangle Based analysis) cannot prove a better-than-$16/7$ approximation.

In \Cref{sec:analysis}, we pick a particular algorithm from this framework (\Cref{alg:pivot-supernodes}), that instantiates the probability $p$ as a convex combination of $f^+,f^-$ (defined in \Cref{sec:prelims}).

\begin{algorithm2e}[!ht]
	 	\DontPrintSemicolon
	 	\SetKwInOut{Input}{Input}
	 	\SetKwInOut{Output}{Output}
	
	 	\caption{Strict-Pivoting Framework}
	 	\label{alg:strictPivot}
	
	 	\Input{Instance $G = (V,E^+\uplus E^-,F,H)$}
	 	\Output{A clustering $\XC$ of $V$}
	
	 	\BlankLine
	 	Obtain $\{x_{uv}\}_{uv \in \binom{V}{2}}$, a feasible solution to the standard LP (\Cref{fig:LP})\;
	
	 	$\XC \gets \set{\emptyset}$\;
	 	\While{$V\neq \emptyset$}{
		 		Pick a pivot $u \in V$ uniformly at random and let $U = s(u)$ be the supernode containing $u$\;
		 		Initialize a new cluster $C\gets U$\;
                \tcp{The following iterations do not necessarily need to be performed independently for each $W$ (one can allow correlations)}
		 		\ForEach{{supernode} $W\ne U$}{
                    Let $E_{UW}$ be the pairs in $E^+$ with one endpoint in $U$ and the other in $W$\;
                    Let $w$ be an arbitrary node in $W$\;
                    \If{$x_{uw} < 0.5$}{
                        Include $W$ into $C$ with probability $p(x_{uw},|E_{UW}|,|U||W|)$\;         \tcp{$p(\cdot,\cdot,\cdot)$ may differ across different algorithms in this framework}
                    }
                }
		 		Remove $C$ from $V$\;
		 		$\XC \gets \XC \cup \{C\}$\;
		 	}
	
	 	\Return{$\XC$}\;
	 \end{algorithm2e}

\Cref{alg:strictPivot} constructs a clustering of the graph using the fractional LP distances $x_{uv}$. In each iteration, it selects a uniformly random pivot vertex $u$ (contained in supernode $U=s(u)$) from the remaining vertices and forms a new cluster by adding every remaining supernode $W$ (containing some node $w$) with a probability determined by:
\begin{itemize}
    \item the LP value $x_{uw}$ (note that this is independent of the choice of $w$, because all nodes in $W$ have LP distance $0$ to each other, directly from \Cref{fig:LP}),
    \item the number of pairs in $E^+$ with one endpoint in $U$ and the other in $W$, %$V$,
    \item and the number of pairs with one endpoint in $U$ and the other in $W$. %$V$.
\end{itemize}
Different algorithms may determine the inclusion probability of $W$ in different ways, but always as a function of the aforementioned three quantities.
The only restriction is that if $x_{uw}\ge 0.5$, then $W$ is not included in $C$ (hence the term \emph{strict pivoting}).

We prove that the clustering output by a strict pivoting algorithm respects all the hard constraints.

\begin{lemma} \label{lem:satisfiesAllConstraints}
The output of a strict pivoting algorithm always satisfies all friendly and hostile constraints.
\end{lemma}
\begin{proof}
By construction, no supernode is ever split, therefore the friendly constraints are satisfied.
Regarding hostile constraints, notice that the LP distance of any node in a cluster $C$ to the pivot of $C$ is less than $0.5$; by triangle inequality the LP distance between any two nodes in a cluster $C$ is less than one, which means that they are not hostile (directly from \Cref{fig:LP}).
\end{proof}

Notice that our framework is quite general.
For example, it does not force the inclusion of each $W$ to be decided independently; one could use correlated rounding. Additionally, the feasible solution to the standard LP could be obtained through solving the more restricted Sherali-Adams relaxation.
Yet, we show that a straightforward Triangle-Based analysis cannot yield a better-than-$\frac{16}{7}$ approximation.

It is important to note that our lower bound does not mean that no Strict-Pivoting algorithm can obtain a better-than-$\frac{16}{7}$ approximation.
It only shows that the standard Triangle-Based analysis is not strong enough to prove such a bound.
In fact, in \Cref{sec:beyond} we present a Strict-Pivoting algorithm that obtains a better-than-$\frac{16}{7}$ approximation, using a more sophisticated analysis.

\paragraph{Analysis} We now present the standard analysis for such algorithms.
We first argue that it is sufficient to analyze a single, arbitrary iteration in \Cref{alg:strictPivot} (as noted in \cite{sub2}).
Suppose in iteration $m$ we create the cluster $C$.
Let $\ALG_m$ denote the expected cost incurred in iteration $m$ and $\mathsf{LP}_m$ be the expected amount of LP value removed by this iteration, that is the sum of LP contributions ($x_{uv}$, in case of positive edges, $1-x_{uv}$ otherwise) of pairs with at least one endpoint in $C$ and the other endpoint outside $C$, if it exists, unclustered.
If we can establish that $\ALG_m \leq \alpha \cdot \mathsf{LP}_m$ holds for all $m$, then the overall approximation ratio follows by summing over all $R$ iterations of the algorithm:
$$
\expected{\ALG} = \expected{\sum_{m=0}^{R} \ALG_m} \leq \alpha \cdot \expected{\sum_{m=0}^{R} \mathsf{LP}_m} = \alpha \cdot \mathsf{LP},
$$
where $\ALG$ is the total expected cost of the algorithm and $\mathsf{LP}$ denotes the objective value of the LP. This reduces the problem to analyzing a single, arbitrary iteration, for which we drop the subscript $m$ for convenience.

We now define the following quantities. For a given pivot $u$, let \cost{u}{v}{w} be the probability that the edge $vw$ is violated ($vw \in E^+$ but they end up in different clusters, or $vw\in E^-$ but they end up in the same cluster), and let \lp{u}{v}{w} be the LP contribution of $vw$ (i.e., $x_{vw}$ if $vw$ is positive and $1 - x_{vw}$ if it is a negative edge) times the probability that $vw$ is decided (i.e., that either $v$ or $w$ joins $C$).

We call a set of three distinct vertices a \emph{triangle} and a set of two vertices a \emph{degenerate triangle}. We aggregate the above notations for both triangles and degenerate triangles. For a triangle $T=uvw$ we define
\begin{align*}
	& \Cost{T} = \COST{u}{v}[w] = \cost{u}{v}{w} + \cost{v}{u}{w} + \cost{w}{u}{v} \\
	& \Lp{T} = \LP{u}{v}[w] = \lp{u}{v}{w} + \lp{v}{u}{w} + \lp{w}{u}{v}.
\end{align*}
For a degenerate triangle $uv$, we define
\begin{align*}
	& \COST{u}{v} = \cost{u}{u}{v} + \cost{v}{u}{v} \\
	& \LP{u}{v} = \lp{u}{u}{v} + \lp{v}{u}{v}.
\end{align*}

\begin{lemma} \label{lem:triangle}
The approximation factor of an algorithm in the framework of \Cref{alg:strictPivot} is at most 
\[\frac{\sum_{uvw\in\binom{V}{3}} \COST{u}{v}[w] + \sum_{uv\in\binom{V}{2}} \COST{u}{v}}{\sum_{uvw\in\binom{V}{3}} \LP{u}{v}[w] + \sum_{uv\in\binom{V}{2}} \LP{u}{v}}.\]
\end{lemma}
\begin{proof}
It is straightforward to see that for a single arbitrary iteration
\begin{align*}
	\ALG & = \mathbb{E}_{u\in V} \left[\sum_{vw \in \binom{V}{2}} \cost{u}{v}{w}\right] \\
	\mathsf{LP} & = \mathbb{E}_{u\in V} \left[\sum_{vw \in \binom{V}{2}} \lp{u}{v}{w}\right].
\end{align*}
Hence, it suffices to upper bound:
$$
\frac{\ALG}{\mathsf{LP}} = \frac{\mathbb{E}_{u\in V} \left[\sum_{vw \in \binom{V}{2}} \cost{u}{v}{w}\right]}{\mathbb{E}_{u\in V} \left[\sum_{vw \in \binom{V}{2}} \lp{u}{v}{w}\right]} = \frac{\sum_{uvw\in\binom{V}{3}} \COST{u}{v}[w] + \sum_{uv\in\binom{V}{2}} \COST{u}{v}}{\sum_{uvw\in\binom{V}{3}} \LP{u}{v}[w] + \sum_{uv\in\binom{V}{2}} \LP{u}{v}}. 
$$
\end{proof}

\begin{corollary} \label{cor:triangleBased}
Let $\rho$ be such that $\frac{\Cost{T}}{\Lp{T}} \le \rho$ for every %any 
(possibly degenerate) triangle $T$ w.r.t.\ a
feasible solution to the LP in \Cref{fig:LP}.
Then the approximation factor of an algorithm in the framework of \Cref{alg:strictPivot} is upper bounded by $\rho$.
\end{corollary}

We now show that using \Cref{cor:triangleBased} (that is, the standard Triangle Based analysis) we cannot prove any better-than-$\frac{16}{7}$ approximating for any strict pivoting algorithm (\Cref{alg:strictPivot}).

\begin{lemma}%\label{lem:barrier}
For any strict pivoting algorithm (\Cref{alg:strictPivot}) there exist a feasible solution $\set{x_{uv}}_{uv\in \binom{V}{2}}$ to the standard LP in \Cref{fig:LP} and a triangle $T$ such that $\frac{\Cost{T}}{\Lp{T}} \ge \frac{16}{7}$.
\end{lemma}
\begin{proof}
Let $T=tuv$ be a triangle such that $t$ has two plus incident edges $tu$ and $tv$ with $x_{tu} = x_{tv} = 0.25$, while $x_{uv} = 0.5$. 
Assume without loss of generality that there are no friendly constraints, therefore each node is in a singleton supernode ($|s(t)| = |s(u)| = |s(v)| = 1$).
We show that either when $uv\in E^+$ or when $uv\in E^-$, we get $\frac{\Cost{T}}{\Lp{T}} \ge \frac{16}{7}$ (with one minor exception which is handled independently).

Let $q = p(0.25,1,1)$, where $p(\cdot, \cdot, \cdot)$ depends on the strict pivoting algorithm we focus on (see \Cref{alg:strictPivot}), and $\alpha$ be an upper bound on $\frac{\Cost{T}}{\Lp{T}}$.
Notice that when $u$ (resp.\ $v$) is the pivot, then $v$ (resp.\ $u$) is not included in the cluster, directly by \Cref{alg:strictPivot} since $x_{uv} = 0.5$.

	\begin{enumerate}
        \item If $q=0$ then for the degenerate triangle $tu$ we get $\LP{t}{u} = 2\cdot 0.25, \COST{t}{u} = 2\cdot 1$, therefore
        \[\frac{\COST{t}{u}}{\LP{t}{u}} = 4 \ge \frac{16}{7}.\]
        \item For $q > 0$, consider a (non-degenerate) triangle $tuv$.

        Since a strict pivoting algorithm could use correlated rounding, for a plus edge $ab$ such that $x_{ab} = 0.25$, the value $q = p(0.25,1,1)$ is the marginal probability of $a$ being added to the cluster given that $b$ is the pivot.
        Let  $r$ denote $\Pr[uv \text{ is decided} | t \text{ is pivot}]$.

        Notice that
        \begin{align*}
            2q = &\Pr[\text{exactly one of }u,v \text{ is added to cluster} | t \text{ is pivot}] \\+ 2&\Pr[\text{both of }u,v \text{ are added to cluster} | t \text{ is pivot}].
        \end{align*}

        and similarly

        \begin{align*}
            r = &\Pr[\text{exactly one of }u,v \text{ is added to cluster} | t \text{ is pivot}] \\+ &\Pr[\text{both of }u,v \text{ are added to cluster} | t \text{ is pivot}].
        \end{align*}
        
        \begin{enumerate}
		\item Let $uv$ be a plus edge. In this case, one can see that
        $$\cost{t}{u}{v} =  \Pr[\text{exactly one of }u,v \text{ is added to cluster} | t \text{ is pivot}] = 2r - 2q.$$ 

        Using the above,
		\begin{align*}
			\LP{t}{u}[v] & = \lp{u}{t}{v} + \lp{t}{u}{v} + \lp{v}{t}{u}
			= r \cdot 0.5 + 2q \times 0.25 \\
            & = 0.5 \cdot (q + r),\\
			\COST{t}{u}[v] & = \cost{u}{t}{v} + \cost{t}{u}{v} + \cost{v}{t}{u}
			= 2 (r - q) + 2q.\\
            & = 2r.
		\end{align*}
		
		By definition of $\alpha$ we have $\alpha \cdot \LP{t}{u}[v] \geq \COST{t}{u}[v]$, therefore
		$
		\frac{\alpha}{2} \cdot (q + r) \geq 2r
		$.
        
		By rearranging we get $\frac{\alpha}{2} q \geq (2 - \frac{\alpha}{2}) r$
		\begin{align*}
		  \Longrightarrow \frac{r}{q} \leq \frac{\alpha}{4 - \alpha}.
		\end{align*}
		\item Let $uv$ be a minus edge. In this case, $$\cost{t}{u}{v} = \Pr[\text{both of }u,v\text{ are added to cluster}| t \text{ is pivot}] = 2q - r.$$

        Using the above,
		\begin{align*}
			\LP{t}{u}[v] & = \lp{u}{t}{v} + \lp{t}{u}{v} + \lp{v}{t}{u} = r \times 0.5 + 2q \times 0.25 \\
            & = 0.5 \cdot (q + r),\\
			\COST{t}{u}[v] & = \cost{u}{t}{v} + \cost{t}{u}{v} + \cost{v}{t}{u} = (2q - r) + 2q\\
            & = 4q - r.
		\end{align*}
		Again by rearranging we must have $\frac{\alpha}{2} \cdot (q + r) \geq 4q - r$
		\begin{align*}
			\Longrightarrow \frac{r}{q} \geq \frac{8 - \alpha}{\alpha + 2}.
		\end{align*}
        \end{enumerate}
	\end{enumerate}
	We conclude that if $q>0$, then
	$$
	\frac{\alpha}{4 - \alpha} \geq \frac{r}{q} \geq \frac{8 - \alpha}{\alpha + 2}.
	$$
	For $\alpha < 4$ all denominators are positive, and we get
	\begin{align*}
	(4 - \alpha) (8 - \alpha) & \leq \alpha (\alpha + 2)\\
	\alpha^2 - 12 \alpha + 32 & \leq \alpha^2 + 2\alpha \\
    \frac{16}{7} & \leq \alpha.
	\end{align*}
	which is the desired result.
\end{proof}

\section{\texorpdfstring{$16/7$ Approximation}{16/7 Approximation}} \label{sec:analysis}
Our algorithm (\Cref{alg:pivot-supernodes}) is a strict pivoting algorithm (as defined in \Cref{alg:strictPivot}).
In particular, notice that if a supernode $W$ has LP distance at least $0.5$ to the pivot, then the probability of $W$ being included in the cluster $C$ is $0$, by definition of $f^+,f^-$.
Therefore the output of \Cref{alg:pivot-supernodes} satisfies all hard constraints, by \Cref{lem:satisfiesAllConstraints}.

\Cref{alg:pivot-supernodes} processes each supernode independently.
Using the notation of \Cref{alg:strictPivot}, when the pivot is $u$, belonging to a supernode $U=s(u)$, and we process supernode $W\ni w$, we let $\beta_{UW} = \frac{|E_{UW}|}{|U||W|}$ be the fraction of plus edges over the total pairs between the two supernodes.
Then we use 
\[p(x_{uw},|E_{UW}|,|U||W|) \stackrel{\mathrm{def}}{=} \frac{|E_{UW}|}{|U||W|} f^+(x_{uw}) + \frac{|U||W|-|E_{UW}|}{|U||W|}f^-(x_{uw}) = \beta_{UW}f^+(x_{uw}) + (1-\beta_{UW})f^-(x_{uw}).\]

Note that by definition of supernodes, all nodes inside a supernode have LP distance $0$ to each other. Therefore, the choice of $w\in W$ in the above expression does not matter, as all pairs with one endpoint in $U$ and the other in $W$ have the same LP distance (by triangle inequality).

%%%%%%%%%%%%%%% SUPERNODES %%%%%%%%%%%%%%%%%%%

\begin{algorithm2e}[!ht]
	 	\DontPrintSemicolon
	 	\SetKwInOut{Input}{Input}
	 	\SetKwInOut{Output}{Output}
	
	 	\caption{The $\frac{16}{7}-\gamma$ approximation}
	 	\label{alg:pivot-supernodes}
	
	 	\Input{Instance $G = (V,E^+\uplus E^-,F,H)$}
	 	\Output{A clustering $\XC$ of $V$}
	
	 	\BlankLine
	 	Obtain $\{x_{uv}\}_{uv \in \binom{V}{2}}$, a feasible solution to the standard LP (\Cref{fig:LP})\;
	
	 	$\XC \gets \set{\emptyset}$\;
	 	\While{$V\neq \emptyset$}{
		 		Pick a pivot $u \in V$ uniformly at random and let $U = s(u)$ be the supernode containing $u$\;
		 		Initialize a new cluster $C\gets U$\;
		 		\ForEach{{supernode} $W\ne U$}{
                    Let $E_{UW}$ be the pairs in $E^+$ with one endpoint in $U$ and the other in $W$\;
                    Let $w$ be an arbitrary node in $W$\;
                    Let $\beta_{UW} = \frac{|E_{UW}|}{|U||W|}$\;
                    Include $W$ into $C$ with probability $\beta_{UW} f^+(x_{uw}) + (1 - \beta_{UW}) f^-(x_{uw})$\; 
                }
		 		Remove $C$ from $V$\;
		 		$\XC \gets \XC \cup \{C\}$\;
		 	}
	
	 	\Return{$\XC$}\;
	 \end{algorithm2e}

We now analyze $\frac{\Cost{T}}{\Lp{T}}$.
In fact, for technical reasons, we analyze $\frac{\Cost{T}+\delta}{\Lp{T}}$, where $\delta$ is a constant.
We prove that it suffices to upper bound this quantity in the case where the edges connecting any two different supernodes (each containing some of the vertices of $T$) are all of the same type (either all in $E^+$ or all in $E^-$).
In particular, we prove the following.

\begin{lemma}\label{lem:edgesSame}
Let $u,w,z$ be distinct nodes, $\delta$ be a constant, and $\mathcal{T}$ be the set of (non-degenerate) triangles with one endpoint in the supernode $s(u)$, another in $s(w)$, and the other in $s(z)$.
Then the expression
\[\frac{\sum_{T\in \mathcal{T}}(\Cost{T}+\delta)}{\sum_{T\in \mathcal{T}}\Lp{T}}\]
is maximized when for each pair of supernodes in $\set{s(u),s(w),s(z)}$ either all edges between them belong to $E^+$ or they all belong to $E^-$.

Similarly, let $\mathcal{T'}$ be the set of degenerate triangles with one endpoint in $s(u)$ and the other in $s(w)$.
Then
\[\frac{\sum_{T'\in \mathcal{T'}}(\Cost{T'}+\delta)}{\sum_{T'\in \mathcal{T'}} \Lp{T'}}\]
is maximized when either all edges between $s(u)$ and $s(v)$ belong to $E^+$ or they all belong to $E^-$.
\end{lemma}

\begin{proof}

\textbf{Case: $T=uwz$ and all belong to two distinct supernodes ($s(u)\ne s(w)=s(z)$): }
Let $U=s(u)$ be the supernode containing $u$ and $W=s(w)=s(z)$ the supernode containing $w,z$.
In this case, if $u$ is the pivot, then either $w,z$ both end up in the cluster with $u$, or neither do.
As $w,z$ are in the same supernode, we have $wz\in E^+$ and $x_{wz}=0$, therefore $\cost{u}{w}{z} = \lp{u}{w}{z} = 0$.
If $w$ is the pivot, then $z$ ends up in the cluster with $w$, and $u$ ends up in the same cluster with probability $\left(\beta_{UW}f^+(x_{uz}) + (1-\beta_{UW})f^{-}(x_{uz})\right)$.
Therefore, if $uz\in E^+$ then $\cost{w}{u}{z} = \left(\beta_{UW}(1-f^+(x_{uz})) + (1-\beta_{UW})(1-f^{-}(x_{uz}))\right)$, while if $uz\in E^-$ then $\cost{w}{u}{z} = \left(\beta_{UW}f^+(x_{uz}) + (1-\beta_{UW})f^{-}(x_{uz})\right)$.
By an analogous argument, we get identical expressions for the case when $z$ is the pivot. 
Thus, 
\begin{align*}
\sum_{u\in U, wz\in \binom{W}{2}} \Cost{uwz} &= (|W| -1)\cdot [|E_{UW}| \cdot \left(\beta_{UW}(1-f^+(x_{uz})) + (1-\beta_{UW})(1-f^{-}(x_{uz}))\right)\\&+ (|U||W| - |E_{UW}|) \cdot \left(\beta_{UW}f^+(x_{uz}) + (1-\beta_{UW})f^{-}(x_{uz})\right)].
\end{align*}
The above simplifies to
\[\frac{|W| - 1}{|U||W|} \cdot \left[(2f^{-}(x_{uz}) - 2f^+(x_{uz}))\beta_{UW}^2 + (1 + f^+(x_{uz}) -3f^-(x_{uz}))\cdot \beta_{UW} + f^-(x_{uz})\right],\]
which is convex in $\beta_{UW}$ since the coefficient of $\beta_{UW}^2$ is nonnegative (by~\Cref{lemma:funcsInequality}). 

Now, if $w$ is the pivot, then $z$ is included in its cluster, therefore $\lp{w}{u}{z} = x_{uz}$ if $uz\in E^+$ and $1-x_{uz}$ otherwise. An analogous argument holds also for the case when $z$ is the pivot.
Thus,
\begin{align*}
    \sum_{u \in U, wz \in \binom{W}{2}} \Lp{uwz} = \frac{|W| - 1}{|U||W|} \cdot [\beta_{UW} \cdot x_{uz} + (1 - \beta_{UW}) \cdot(1 - x_{uz})].
\end{align*} 
is linear in $\beta_{UW}$.
Hence, in order to check whether 
\[\sum_{u\in U, wz\in \binom{W}{2}} \Cost{uvw} - \alpha \sum_{u\in U, wz\in \binom{W}{2}} \Lp{uvw} + \delta \leq 0,\]
for all $\beta_{UW} \in [0,1]$,
one simply needs to determine that the maximum value of the convex (in $\beta_{UW}$) function
$$\sum_{u\in U, wz\in \binom{W}{2}} \Cost{uvw} - \alpha \sum_{u\in U, wz\in \binom{W}{2}} \Lp{uvw} + \delta$$
is non-positive. In particular, this amounts to considering only the cases that $\beta_{UW}\in \set{0,1}$.

This means that if $u,w,z$ are such that $s(u) \neq s(w) = s(z)$, we can upper bound $\frac{\COST{u}{w}[z] + \delta}{\LP{u}{w}[z]}$ by upper bounding the cases where all edges between $s(u)$ and $s(w)$ are minus (corresponding to $\beta_{UW} = 0$) and all edges between $s(u)$ and $s(w)$ are plus (corresponding to $\beta_{UW} = 1$).

\noindent \textbf{Case: $T=uwz$ and all belong to distinct supernodes ($|\{s(u),s(w),s(z)\}| = 3$): }
Consider an arbitrary iteration of the main while loop of \Cref{alg:pivot-supernodes}. Let $U$ be the supernode containing the pivot $u$. Let $W,Z\ne U$ be distinct supernodes in that iteration.

For $w \in W, z \in Z$ such that $wz \in E^+$, we have $\cost{u}{w}{z}$ is the probability that exactly one of $w,z$ is included in $C$ and the other is not.
\begin{align*}
X_{UWZ} \stackrel{\mathrm{def}}{=} & \left(\beta_{UW}f^+(x_{uw}) + (1-\beta_{UW})f^{-}(x_{uw})\right)\cdot \left(\beta_{UZ}(1-f^+(x_{uz})) + (1-\beta_{UZ})(1-f^{-}(x_{uz}))\right) \\ &+ \left(\beta_{UW}(1-f^+(x_{uw})) + (1-\beta_{UW})(1-f^{-}(x_{uw}))\right)\cdot \left(\beta_{UZ}f^+(x_{uz}) + (1-\beta_{UZ})f^{-}(x_{uz})\right).
\end{align*}

Similarly, for $w \in W, z\in Z$ such that $wz \in E^-$, we have $\cost{u}{w}{z}$ is the probability that both the supernodes $W, Z$ are included in the cluster formed with $U$ as pivot, which is 
\[Y_{UWZ} \stackrel{\mathrm{def}}{=} \left(\beta_{UW}f^+(x_{uw}) + (1-\beta_{UW})f^{-}(x_{uw})\right)\cdot \left(\beta_{UZ}f^+(x_{uz}) + (1-\beta_{UZ})f^{-}(x_{uz})\right).\]

It is clear from the above expressions that $\cost{u}{w}{z}$ depends only on whether $wz\in E^+$ or $wz\in E^-$, hence: 
\begin{align*}
    \sum_{w \in W, z\in Z} \cost{u}{w}{z} = |W||Z| \left[\beta_{WZ} X_{UWZ} + (1-\beta_{WZ})Y_{UWZ}\right].
\end{align*}

The above sum is a multilinear polynomial in $\beta_{UW}, \beta_{UZ}, \beta_{WZ}$. 

Let $A_{UWZ}$ denote the probability that at least one of $w,z$ is included in $C$, that is
\begin{align*}
    A_{UWZ} \stackrel{\mathrm{def}}{=} & (\beta_{UW}f^+(x_{uw}) + (1-\beta_{UW})f^{-}(x_{uw}))\cdot (\beta_{UZ}f^+(x_{uz}) + (1-\beta_{UZ})f^{-}(x_{uz}))\\
    &+(\beta_{UW}f^+(x_{uw}) + (1-\beta_{UW})f^{-}(x_{uw}))\cdot (\beta_{UZ}(1-f^+(x_{uz})) + (1-\beta_{UZ})(1-f^{-}(x_{uz})))\\
    &+ (\beta_{UW}(1-f^+(x_{uw})) + (1-\beta_{UW})(1-f^{-}(x_{uw})))\cdot (\beta_{UZ}f^+(x_{uz}) + (1-\beta_{UZ})f^{-}(x_{uz})).
\end{align*}

Now, $\lp{u}{w}{z} = A_{UWZ}\cdot x_{wz}$ if $wz$ is a plus edge and $\lp{u}{w}{z} = A_{UWZ}\cdot (1-x_{wz})$ otherwise. Hence, $\lp{u}{w}{z}$ depends only on whether the edge $wz\in E^+$ or $wz\in E^-$.

Thus, 
\begin{align*}
    \sum_{w \in W, z \in Z} \lp{u}{w}{z} = |W||Z| \cdot A_{UWZ}\cdot (\beta_{WZ}x_{wz} + (1-\beta_{WZ})(1-x_{wz})),
\end{align*}
which is also multilinear in $\beta_{UW},\beta_{UZ},\beta_{WZ}$.

Hence, $\sum_{u\in U, w \in W, z \in Z} \Cost{uwz}$ and $\sum_{u\in U, w \in W, z \in Z} \Lp{uwz}$ are also  multilinear in $\beta_{UW},\beta_{UZ},\beta_{WZ}$. Therefore, an inequality of the form 
\begin{align*}
    \sum_{u\in U, w \in W, z \in Z} \Cost{uwz} \leq \alpha \cdot \sum_{u\in U, w \in W, z \in Z} \Lp{uwz} - \delta
\end{align*}
is satisfied if it is satisfied for $\beta_{UW},\beta_{UZ},\beta_{WZ} \in \{0,1\}$.

This means that if $u,w,z$ are all in different supernodes, we can upper bound $\frac{\COST{u}{w}[z] + \delta}{\LP{u}{w}[z]}$ by upper bounding the cases where, out of $\beta_{UW},\beta_{UZ},\beta_{WZ} \in \{0,1\}$ (i) all are $0$, (ii) one is $1$ and two are $0$, (iii) two are $1$ and one is $0$, (iv) all three are $1$.
\\~\\
\textbf{Case: $T=uw$ and they belong to distinct supernodes ($s(u)\ne s(w)$): }
Let $U=s(u)$ be the supernode containing $u$ and $W=s(w)$ the supernode containing $w$.
If $u$ is the pivot, then $w$ ends up in the same cluster with probability $\left(\beta_{UW}f^+(x_{uw}) + (1-\beta_{UW})f^{-}(x_{uw})\right)$.
Therefore, if $uw\in E^+$ then $\cost{u}{u}{w} = \left(\beta_{UW}(1-f^+(x_{uw})) + (1-\beta_{UW})(1-f^{-}(x_{uw}))\right)$, while if $uz\in E^-$ then $\cost{u}{u}{w} = \left(\beta_{UW}f^+(x_{uw}) + (1-\beta_{UW})f^{-}(x_{uw})\right)$.
The expressions for $\cost{w}{u}{w}$ are identical.

As the pivot always ends up in the cluster, we get that $\lp{u}{u}{w} = \lp{w}{u}{w} = x_{uw}$ if $uw\in E^+$ and $1-x_{uw}$ otherwise.

Again, we get that to prove an expression of the form:
\[\sum_{u\in U, w\in W} \Cost{uw} \le \alpha \sum_{u\in U, w\in W} \Lp{uw} - \delta\]
it suffices to prove it for $\beta_{UW}\in \set{0,1}$.
\end{proof}

\subsection{Upper Bounding each type of Triangle}
In the rest of this section, we always assume that for any two supernodes, the edges between them are all of the same type (either all in $E^+$ or all in $E^-$).
This type of analysis suffices, because of \Cref{lem:edgesSame}.

Therefore, if $u$ is the pivot and $v$ is a node in a different supernode than $u$, then the probability of $v$ being included in the same cluster with $u$ is $f^+(uv)$ if $uv\in E^+$, and $f^-(uv)$ if $uv\in E^-$.
We now analyze each different type of triangle.

For technical reasons, relating to breaking beyond the $\frac{16}{7}$ barrier (\Cref{sec:beyond}), in what follows we do not only analyze the quantity $\frac{\COST{u}{w}[z]}{\LP{u}{w}[z]}$.
We also analyze cases where the numerator is slightly decreased or slightly increased, for particular cases of triangles.
For this reason we introduce the small positive constants $\gamma=10^{-6}, \slackBad{}=5\times 10^{-3}, \charge{}=10^{-6}$.

We first prove a helpful lemma we are using throughout the paper:

\begin{lemma} \label{lem:enoughBudget}
Let $uvw$ be in distinct supernodes, and $x_{uv} \in 0.25 \pm \slackBad{}$.
Then $\lp{u}{u}{v} > 0.2$ and $\LP{u}{v} > 0.4$.
Furthermore, if $x_{uw}\in 0.25 \pm \slackBad{}$, then $\lp{v}{u}{w} > 0.1$ and $\LP{u}{v}[w] > 0.2$.
\end{lemma}
\begin{proof}
$\lp{u}{u}{v}$ is equal to $x_{uv}$ if $uv\in E^+$, and $1-x_{uv}$ otherwise; it follows that $\lp{u}{u}{v} \ge 0.25-\slackBad{} > 0.2$.
As $\lp{u}{u}{v} = \lp{v}{u}{v}$, we get $\LP{u}{v} > 0.4$.

Furthermore, if $v$ is the pivot, we have that the probability of $u$ being in the cluster of $v$ is at least $0.5$ (by definition of $f^+$ and $f^-$), and therefore $\lp{v}{u}{w} \ge 0.5 (0.25-\slackBad{}) > 0.1$.
The same holds for $\lp{w}{u}{v}$, and thus $\LP{u}{v}[w] > 0.2$.
\end{proof}

Now we show that it suffices to only analyze triangles where all their nodes belong to different supernodes.
\begin{lemma}\label{lemma:degenerate-bound-1}
If the nodes of a triangle $T$ belong to at most two distinct supernodes, then
	$$
	\frac{\Cost{T}}{\Lp{T}} \leq 2.
	$$

    Furthermore, if $uv\subseteq T$ and $x_{uv} \in 0.25 \pm \slackBad{}$, then $\frac{\Cost{T} + 9\charge{}}{\Lp{T}} \leq 2.1$.
\end{lemma}
\begin{proof}
If the nodes of a triangle $T$ all belong to the same supernode, then they all end up in the same cluster.
Furthermore they are all connected by edges in $E^+$, therefore $\Cost{T} = 0$.
Notice that in this case the LP distances $x$ are all $0$, therefore we cannot have $x_{uv} \in 0.25 \pm \slackBad{}$.

If $T=uv$ (degenerate triangle) then by definition we have:
	\begin{align*}
		\COST{u}{v} & = \cost{u}{u}{v} + \cost{v}{u}{v}, \\
		\LP{u}{v} & = \lp{u}{u}{v} + \lp{v}{u}{v}.
	\end{align*}
	By symmetry, the contribution when $u$ is the pivot is the same as when $v$ is the pivot, so $\frac{\COST{u}{v}}{\LP{u}{v}} = \frac{\cost{u}{u}{v}}{\lp{u}{u}{v}}$. The term $\cost{u}{u}{v}$ is the probability of violating edge $uv$ when $u$ is the pivot. The term $\lp{u}{u}{v}$ is the LP cost of $uv$, since the probability of the edge $uv$ being decided is $1$. We proceed with a case analysis on the sign of the edge.
	
	\begin{enumerate}
		\item[(i)]  $uv$ is a negative edge. The LP cost is $1 - x_{uv}$. A violation occurs if $v$ is clustered with $u$, which happens with probability $f^-(x_{uv})$.
        The ratio $\frac{\COST{u}{v}}{\LP{u}{v}}$ is thus $\frac{f^-(x_{uv})}{1 - x_{uv}}$. For $x_{uv}\in [0, 0.5)$, the definition $f^-(x_{uv}) = 1 - x_{uv}$ makes the ratio exactly $1$.
        For $x_{uv}\in [0.5, 1]$, the definition $f^-(x_{uv}) = 0$ gives $\frac{\COST{u}{v}}{\LP{u}{v}} = 0$.
		\item[(ii)]  $uv$ is a positive edge. The LP cost is $x_{uv}$. A violation occurs if $v$ is not clustered with $u$, which happens with probability $1 - f^+(x_{uv})$, making the ratio $\frac{\COST{u}{v}}{\LP{u}{v}} = \frac{1 - f^+(x_{uv})}{x_{uv}}$.
        For $x_{uv}\in [0, 0.5)$, using \Cref{lemma:func-ineq} we have $f^+(x_{uv}) \geq 1 - \frac{3}{2} x_{uv}$, which implies the ratio  is at most $\frac{3}{2}$.
        For $x_{uv}\in [0.5, 1]$, the definition $f^+(x_{uv}) = 0$ makes the ratio $\frac{1}{x_{uv}}$ which is maximized at $x_{uv} = 0.5$, yielding $\frac{\COST{u}{v}}{\LP{u}{v}} \le 2$.
	\end{enumerate}

    Finally, if $x_{uv}\in 0.25 \pm \slackBad{}$, then $\frac{\COST{u}{v} + 9\charge{}}{\LP{u}{v}} \le 2 + \frac{9}{0.4}\charge \le 2.1$, by \Cref{lem:enoughBudget}.

Now if $T=uvw$ and $s(u) \ne s(v) = s(w)$, we get that $vw\in E^+$.
If $u$ is the pivot then either both or none of $vw$ are included in the cluster of $u$, meaning that $\cost{u}{v}{w} = 0$.
If $v$ is the pivot, $w$ is included in its cluster, while $u$ is included with probability $f^+(x_{uv})$ if $uv\in E^+$, and $f^-(x_{uv})$ otherwise. The expressions are identical when $w$ is the pivot because $x_{uw} = x_{uv}$ and both $uv,uw$ are of the same type (either both in $E^+$ or both in $E^-$).
Therefore $\frac{\Cost{T}}{\Lp{T}} \le 2$, exactly as in the case of a degenerate triangle.

    Finally, if $x_{uv}=x_{uw}\in 0.25 \pm \slackBad{}$, then $\frac{\Cost{T} + 9\charge{}}{\Lp{T}} \le 2 + \frac{9}{0.2}\charge \le 2.1$, by \Cref{lem:enoughBudget}.
\end{proof}

Therefore, for the rest of this section we only consider triangles where all nodes belong to distinct supernodes.

\subsubsection{\texorpdfstring{$---$ Triangles}{--- Triangles}}
\begin{lemma}\label{lemma:---_triangle}
	For a $---$ triangle $T=uvw$,
	$$
	\frac{\Cost{T}}{\Lp{T}} \leq 1.
	$$

    Furthermore, if $x_{uv},x_{uw} \in 0.25 \pm \slackBad{}$, then $\frac{\Cost{T} + 9\charge{}}{\Lp{T}} \leq 1.1$.
\end{lemma}

\begin{proof}
	We first compute the \Cost{} and \Lp{} contributions for each of the three pivots. Let us consider the case where $u$ is the pivot. A violation of the edge $vw$ occurs if and only if both of $v$ and $w$ are included in the cluster seeded by $u$.
    This happens with probability 
	$$\cost{u}{v}{w} = \incl{-}{u}{v}\incl{-}{u}{w}.$$
	The corresponding \Lp{} value of the edge is $1 - x_{vw}$, multiplied by the probability that it is decided, which is 
	$$\lp{u}{v}{w} = (1 - x_{vw}) \Bigl(1 - (1 - \incl{-}{u}{v}) (1 - \incl{-}{u}{w})\Bigr).$$ 
	The contributions for pivots $v$ and $w$ are defined symmetrically. 
%	\begin{align*}
%		\cost{u}{v}{w} & = f^-(x_{uv}) f^-(x_{uw}) \\
%		\lp{u}{v}{w} & = (1 - x_{vw}) \Bigl(1 - (1 - f^-(x_{uv})) (1 - f^-(x_{uw}))\Bigr) \\
%		\cost{v}{w}{u} & = f^-(x_{vw}) f^-(x_{vu}) \\
%		\lp{v}{w}{u} & = (1 - x_{wu}) \Bigl(1 - (1 - f^-(x_{vw})) (1 - f^-(x_{vu}))\Bigr) \\
%		\cost{w}{u}{v} & = f^-(x_{wu}) f^-(x_{wv}) \\
%		\lp{w}{u}{v} & = (1 - x_{uv}) \Bigl(1 - (1 - f^-(x_{wu})) (1 - f^-(x_{wv}))\Bigr).
%	\end{align*}
	To prove the lemma, we must show that the sum of \Cost{} terms is no more than the sum of \Lp{} terms:
	$$
	\sum_{u} f^-(x_{uv}) f^-(x_{uw}) \leq \sum_{u} (1 - x_{vw}) \Bigl(1 - (1 - f^-(x_{uv})) (1 - f^-(x_{uw}))\Bigr),
	$$
	where the sum is cyclic over $uvw$.
	We consider three exhaustive scenarios: (i) at least two of $x_{uv}$, $x_{uw}$ and $x_{vw}$ are at least 0.5, (ii) exactly one of these values is at least 0.5, and (iii) all three values are less than 0.5.
	
	For case (i), without loss of generality, assume $x_{uv} \geq 0.5$ and $x_{uw} \geq 0.5$. By definition, $\incl{-}{u}{v} = \incl{-}{u}{w} = 0$, which implies that all three terms on the left-hand side of the inequality are zero. The inequality holds trivially as the right-hand side is non-negative.
	
	For case (ii), without loss of generality, assume $x_{uw} \geq 0.5$, while $x_{uv} < 0.5$ and $x_{vw} < 0.5$. Then since $f^-(x_{uw}) = 0$, the left-hand side is simplified to $\incl{-}{u}{v} \incl{-}{v}{w}$. The right-hand side is a sum of three non-negative terms. Since the terms for pivots $v$ and $w$ are non-negative, the inequality holds if the LHS is bounded by the term for pivot $u$, which is simplified to $(1 - x_{vw}) \incl{-}{u}{v}$. Thus, it suffices to show
	$$
	\incl{-}{u}{v} \incl{-}{v}{w} \leq (1 - x_{vw}) \incl{-}{u}{v}.
	$$
	As $x_{vw} < 0.5$, we have $\incl{-}{v}{w} = 1 - x_{vw}$, and therefore the required inequality holds with equality.
	
	For case (iii), all three LP values are less than 0.5, so $f^-(x) = 1 - x$ for all relevant terms. The inequality becomes:
\begin{align*}
  &(1 - x_{uv}) (1 - x_{uw}) + (1 - x_{vw}) (1 - x_{uv}) + (1 - x_{vw}) (1 - x_{uw}) \\
  &\leq (1 - x_{vw})(1 - x_{uv}x_{uw}) + (1 - x_{uw})(1 - x_{vw}x_{uv}) + (1 - x_{uv})(1 - x_{vw}x_{uw})
\end{align*}
The third term on the left is at most the first term on the right since $(1 - x_{uw}) \leq (1 - x_{uv} x_{uw})$. Analogous arguments show that the second term on the left is at most the third term on the right and that the third term on the left is at most the second term on the right, and this proves the inequality.
    
	%$$\sum_{w} (1 - x_{uw}) (1 - x_{vw}) \leq \sum_{u} (1 - x_{vw}) (1 - x_{uv} x_{uw}).$$
    Since the inequality is satisfied in all cases, the proof of $\frac{\Cost{T}}{\Lp{T}} \le 1$ is complete.

    Finally, if $x_{uv},x_{uw}\in 0.25 \pm \slackBad{}$, then $\frac{\Cost{T} + 9\charge{}}{\Lp{T}} \le 1 + \frac{9}{0.2}\charge \le 1.1$, by \Cref{lem:enoughBudget}.
\end{proof}

\subsubsection{\texorpdfstring{$--+$ Triangles}{--+ Triangles}}
\begin{lemma}\label{lemma:--+_triangle}
	For a $--+$ triangle $T=uvw$,
	$$
	\frac{\Cost{T}}{\Lp{T}} \leq 2.
	$$
    Furthermore, if $x_{uv},x_{uw} \in 0.25 \pm \slackBad{}$, then $\frac{\Cost{T} + 9\charge{}}{\Lp{T}} \leq 2.1$.
\end{lemma}

\begin{proof}
	Without loss of generality, assume $vw$ is the positive edge, while $uv$ and $uw$ are negative. We first establish the \Cost{} and \Lp{} contributions for each of the three pivots.
	\begin{align*}
		\cost{u}{v}{w} & = (1 - \incl{-}{u}{v}) \cdot \incl{-}{u}{w} + \incl{-}{u}{v} \cdot (1 - \incl{-}{u}{w}), \\
		\lp{u}{v}{w} & = x_{vw} \cdot \Bigl(1 - (1 - \incl{-}{u}{v}) (1 - \incl{-}{u}{w})\Bigr), \\
		\cost{v}{u}{w} & = \incl{-}{u}{v} \incl{+}{v}{w}, \\
		\lp{v}{u}{w} & = (1 - x_{uw})  \Bigl(1 - (1 - \incl{-}{u}{v}) (1 - \incl{+}{v}{w})\Bigr), \\
		\cost{w}{u}{v} & = \incl{-}{u}{w} \incl{+}{v}{w}, \\
		\lp{w}{u}{v} & = (1 - x_{uv})  \Bigl(1 - (1 - \incl{-}{u}{w}) (1 - \incl{+}{v}{w})\Bigr).
	\end{align*}
	Note that the expression for $\cost{u}{v}{w}$ accounts for a violation of the positive edge $(v,w)$ when $u$ is the pivot, which occurs if exactly one of $v$ or $w$ joins the cluster. To prove the lemma, we must show that the sum of \Cost{} terms is at most twice the sum of \Lp{} terms:
%	\begin{align*}
%		\frac{\Cost{T}}{\Lp{T}} & = \frac{(1 - \incl{-}{u}{v}) \cdot \incl{-}{u}{w} + \incl{-}{u}{v} \cdot (1 - \incl{-}{u}{w}) + \incl{-}{u}{v} \incl{+}{v}{w} + \incl{-}{u}{w} \incl{+}{v}{w}}{x_{vw} \cdot \Bigl(1 - (1 - \incl{-}{u}{v}) (1 - \incl{-}{u}{w})\Bigr) + (1 - x_{uw})  \Bigl(1 - (1 - \incl{-}{u}{v}) (1 - \incl{+}{v}{w})\Bigr) + (1 - x_{uv})  \Bigl(1 - (1 - \incl{-}{u}{w}) (1 - \incl{+}{v}{w})\Bigr)}
%	\end{align*}
	$$
	\sum_{u} \cost{u}{v}{w} \leq 2 \cdot \sum_{u} \lp{u}{v}{w},
	$$
	where the sum is cyclic over $uvw$. The proof proceeds by a case analysis on the LP values, which we partition into three exhaustive cases: (i) $x_{vw} \geq 0.5$, (ii) at least one of $x_{uv}$ and $x_{uw}$ is at least 0.5 while $x_{vw} < 0.5$, and (iii) all are less than 0.5.
	\begin{itemize}
	\item For case (i), by definition we have $\incl{+}{v}{w} = 0$. The sum of \Cost{} terms simplifies to:
	$$
		\sum_{u} \cost{u}{v}{w} = \incl{-}{u}{v} + \incl{-}{u}{w} - 2\incl{-}{u}{v}\incl{-}{u}{w}.
	$$
	To bound the sum of \Lp{} terms, we note that the terms are non-negative, so we can lower bound the sum by a single term:
	$$
		\sum_{u} \lp{u}{v}{w} \geq \lp{u}{v}{w} = x_{vw} \cdot \Bigl(\incl{-}{u}{v} + \incl{-}{u}{w} - \incl{-}{u}{v}\incl{-}{u}{w}\Bigr).
	$$
	Since $x_{vw} \geq 0.5$ we have
	\begin{align*}
		\sum_{u} \cost{u}{v}{w} & = \incl{-}{u}{v} + \incl{-}{u}{w} - 2\incl{-}{u}{v}\incl{-}{u}{w} \\
		& \leq \incl{-}{u}{v} + \incl{-}{u}{w} - \incl{-}{u}{v}\incl{-}{u}{w} \\
		& \leq 2x_{vw} \cdot \Bigl(\incl{-}{u}{v} + \incl{-}{u}{w} - \incl{-}{u}{v}\incl{-}{u}{w}\Bigr) \leq 2\cdot\sum_{u} \lp{u}{v}{w}.
	\end{align*}
	
	\item For case (ii), without loss of generality, assume $x_{uv} \geq 0.5$. If $x_{uw}$ is also at least 0.5, then $\incl{-}{u}{v} = \incl{-}{u}{w} = 0$, which causes the sum $\sum_{u} \cost{u}{v}{w}$ to be zero. The inequality holds trivially, as each term in $\sum_{u} \lp{u}{v}{w}$ is non-negative. We therefore proceed assuming $x_{uv} \ge 0.5$ while $x_{uw} < 0.5$ and $x_{vw} < 0.5$. First, we establish a lower bound for the sum of \Lp{} terms. We can bound the term $\lp{w}{u}{v}$ as follows, using the fact that $1 - \incl{+}{v}{w} \leq 1$:
	\begin{align*}
		\lp{w}{u}{v} & = (1 - x_{uv})  \Bigl(1 - (1 - \incl{-}{u}{w}) \cdot (1 - \incl{+}{v}{w})\Bigr) \\
		& \geq (1 - x_{uv}) \Bigl(1 - (1 - \incl{-}{u}{w}) \cdot 1\Bigr) = (1 - x_{uv}) \cdot \incl{-}{u}{w}.
	\end{align*}
	Using this along with $\incl{-}{u}{v} = 0$, we can bound the sum of \Lp{} terms:
	\begin{align*}
		\sum_{u} \lp{u}{v}{w} & = x_{vw} \cdot \incl{-}{u}{w} + (1 - x_{uw}) \cdot \incl{+}{v}{w} + \lp{w}{u}{v} \\
		& \geq \Bigl(1 - (x_{uv} - x_{vw})\Bigr) \cdot \incl{-}{u}{w} + (1 - x_{uw}) \cdot \incl{+}{v}{w}.
	\end{align*}
	Next, by substituting $\incl{-}{u}{v} = 0$, we simplify the sum of \Cost{} terms:
	$$
	\sum_{u} \cost{u}{v}{w} = \incl{-}{u}{w} + \incl{-}{u}{w} \incl{+}{v}{w}.
	$$
	By triangle inequality, we have $x_{uv} - x_{vw} \leq x_{uw}$. Therefore:
	\begin{align*}
		2\cdot \sum_{u} \lp{u}{v}{w} & \geq 2\cdot (1 - x_{uw}) \cdot \incl{-}{u}{w} + 2\cdot (1 - x_{uw}) \cdot \incl{+}{v}{w} \\
		& \geq \incl{-}{u}{w} + \incl{+}{v}{w} \geq \incl{-}{u}{w} + \incl{-}{u}{w}\incl{+}{v}{w} = \sum_{u} \cost{u}{v}{w}.
	\end{align*}
	
	\item For the final case (iii), since \lp{u}{v}{w} is non-negative, it suffices to show
	$$
	2\cdot \Bigl(\lp{v}{u}{w} + \lp{w}{u}{v}\Bigr) \geq \sum_{u} \cost{u}{v}{w}.
	$$
	We first establish a lower bound on the left-hand side. Since $x_{uv} < 0.5$ and $x_{uw} < 0.5$, it follows that $2\cdot (1 - x_{uv}) > 1$ and $2\cdot (1 - x_{uw}) > 1$. This allows us to bound the expression as follows:
	$$
		2\cdot \Bigl(\lp{v}{u}{w} + \lp{w}{u}{v}\Bigr) \geq \incl{-}{u}{v} + \incl{-}{u}{w} + \incl{+}{v}{w}\Bigl(2 - \incl{-}{u}{v} - \incl{-}{u}{w}\Bigr).
	$$
	Therefore, it suffices to prove that the sum of \Cost{} terms is less than or equal to this lower bound. By expanding the definition of \Cost{} terms, this is equivalent to proving:
	\begin{align*}
		\sum_{u} \cost{u}{v}{w} & = \Bigl(\incl{-}{u}{v} + \incl{-}{u}{w}\Bigr) \cdot ( \incl{+}{v}{w} + 1) - 2\cdot \incl{-}{u}{v} \incl{-}{u}{w} \\
		& \leq \incl{-}{u}{v} + \incl{-}{u}{w} + \incl{+}{v}{w}\Bigl(2 - \incl{-}{u}{v} - \incl{-}{u}{w}\Bigr).
	\end{align*}
	Rearranging terms, this inequality simplifies to:
	$$
	2\incl{-}{u}{v} \incl{-}{u}{w} \geq \incl{+}{v}{w} \Bigl(2 \incl{-}{u}{v} + 2 \incl{-}{u}{w} - 2\Bigr).
	$$
	To verify this, let $X = \incl{-}{u}{v}$, $Y = \incl{-}{u}{w}$, and $Z = \incl{+}{v}{w}$. We must prove for all $X, Y, Z \in [0, 1]$ that:
	$$
	2XY \geq Z \cdot (2X + 2Y - 2).
	$$
	The right-hand side is linear in $Z$, so we need to only verify the inequality at the endpoints of the interval $[0, 1]$.
	If $Z = 0$, the inequality becomes $2XY \ge 0$, which holds.
	If $Z = 1$, the inequality becomes 
	$$
	2XY \geq 2X + 2Y - 2,
	$$
	which is equivalent to 
	$$
	2\cdot (1 - X)(1 - Y) \geq 0.
	$$
	This also holds for any $X, Y \in [0, 1]$, which completes the proof for this case.
	\end{itemize}
	The required inequality is satisfied in all cases.

    Finally, if $x_{uv},x_{uw}\in 0.25 \pm \slackBad{}$, then $\frac{\Cost{T} + 9\charge{}}{\Lp{T}} \le 1 + \frac{9}{0.2}\charge \le 2.1$, by \Cref{lem:enoughBudget}.
\end{proof}
\subsubsection{\texorpdfstring{$-++$ Triangles}{-++ Triangles}}
\begin{lemma}\label{lemma:-++_triangle}
	For a $-++$ triangle $T=uvw$,
	$$
	\frac{\Cost{T}}{\Lp{T}} \leq \frac{16}{7}.
	$$

    Furthermore:
    \begin{itemize}
        \item if $x_{uw},x_{vw}\in 0.25\pm \slackBad{}$ and $x_{uv} \in 0.5\pm \slackBad{}$, then $\frac{\Cost{T}-\charge{}}{\Lp{T}} \le \frac{16}{7} - \gamma$,
        \item if $x_{uw},x_{vw}\in 0.25\pm \slackBad{}$ and $x_{uv} \not\in 0.5\pm \slackBad{}$, then $\frac{\Cost{T} + 9\charge{}}{\Lp{T}} \le \frac{16}{7} - \gamma$,
        \item in all other cases $\frac{\Cost{T}}{\Lp{T}} \le \frac{16}{7} - \gamma$.
    \end{itemize}
\end{lemma}

\begin{proof}
	Without loss of generality, assume $uv$ is the negative edge, while $uw$ and $vw$ are positive. We begin by computing the \Cost{} and \Lp{} terms for each of the three pivots.
	\begin{align*}
		\cost{u}{v}{w} & = (1 - \incl{-}{u}{v}) \cdot \incl{+}{u}{w} + \incl{-}{u}{v} \cdot (1 - \incl{+}{u}{w}), \\
		\lp{u}{v}{w} & =  x_{vw} \cdot \Bigl(1 - (1 - \incl{-}{u}{v}) (1 - \incl{+}{u}{w})\Bigr), \\
		\cost{v}{u}{w} & = (1 - \incl{-}{u}{v}) \cdot \incl{+}{v}{w} + \incl{-}{u}{v} \cdot (1 - \incl{+}{v}{w}), \\
		\lp{v}{u}{w} & = x_{uw} \cdot \Bigl(1 - (1 - \incl{-}{u}{v}) (1 - \incl{+}{v}{w})\Bigr), \\
		\cost{w}{u}{v} & = \incl{+}{u}{w} \incl{+}{v}{w}, \\
		\lp{w}{u}{v} & = (1 - x_{uv})  \Bigl(1 - (1 - \incl{+}{u}{w}) (1 - \incl{+}{v}{w})\Bigr).
	\end{align*}
    Note that since $\Lp{T}$ is a sum of 3 terms, each being at most one, it is clearly upper bounded by three.
    Therefore, since $\gamma \leq 3 \cdot \charge{}$, for $x_{uw}, x_{vw} \in 0.25 \pm \slackBad{}$ and $x_{uv} \in 0.5 \pm \slackBad{}$,
    $$
    \frac{\Cost{T} - \charge{}}{\Lp{T}} = \frac{\Cost{T}}{\Lp{T}} - \frac{\charge{}}{\Lp{T}} \leq \frac{\Cost{T}}{\Lp{T}} - \frac{\charge{}}{3} \leq \frac{\Cost{T}}{\Lp{T}} - \gamma.
    $$
    So in this case, it remains to prove that $\frac{\Cost{T}}{\Lp{T}} \leq \frac{16}{7}$. Also if $x_{uw}, x_{vw}\in 0.25 \pm \slackBad{}$, by \Cref{lem:enoughBudget}, $\Lp{T} \geq 0.2$. Therefore if $x_{uv} \notin 0.5 \pm \slackBad{}$,
    $$
    \frac{\Cost{T} + 9\charge{}}{\Lp{T}} = \frac{\Cost{T}}{\Lp{T}} + 9 \frac{\charge{}}{\Lp{T}} \leq \frac{\Cost{T}}{\Lp{T}} + 45 \cdot \charge{}.
    $$
    So in this case it remains to prove $\frac{\Cost{T}}{\Lp{T}} \leq \frac{16}{7} - 45 \cdot \charge{} - \gamma$. In all other cases, we must prove $\frac{\Cost{T}}{\Lp{T}} \leq \frac{16}{7} - \gamma$, thus in each case we must prove that $\frac{\Cost{T}}{\Lp{T}}$ is upper bounded by some constant.
    Let $\alpha$ be a parameter in $[2, \frac{16}{7}]$. We show that:
	$$
	\sum_{u} \cost{u}{v}{w} \leq \alpha \cdot \sum_{u} \lp{u}{v}{w},
	$$
	where the sum is cyclic over $\{u, v, w\}$. The proof proceeds by a case analysis on the LP values, which we partition into three exhaustive cases: (i) at least one of $x_{uw}$ and $x_{vw}$ is at least 0.5, (ii) $x_{uv} \geq 0.5$, and (iii) all three are less than 0.5.
	\begin{itemize}
		\item For case (i), without loss of generality, assume that $x_{uw} \geq 0.5$. If $x_{uv}$ is also at least $0.5$, then $\Cost{T} = \Lp{T} = \incl{+}{v}{w}$, and therefore the desired inequality follows. If $x_{vw} \geq 0.5$, $\incl{+}{u}{w} = \incl{+}{v}{w} = 0$ simplifies the inequality to
		$$
		2\cdot \incl{-}{u}{v} \leq \alpha \cdot (x_{vw} \cdot \incl{-}{u}{v} + x_{uw} \cdot \incl{-}{u}{v}),
		$$
		which clearly holds, because $x_{vw} + x_{uw} \geq 0.5 + 0.5 = 1$ and $\alpha \geq 2$.
		We therefore proceed assuming $x_{uv}, x_{vw} < 0.5$. By definition, $\incl{+}{u}{w} = 0$, which simplifies the sum of \Cost{} terms to
		$$
			\sum_{u} \cost{u}{v}{w} = \incl{-}{u}{v} + \incl{-}{u}{v} \cdot (1 - \incl{+}{v}{w}) + \incl{+}{v}{w} \cdot (1 - \incl{-}{u}{v}).
		$$
		As $x_{uv} < 0.5$ and $x_{uw} \geq 0.5$, it follows that both $1 - x_{uv}$ and $x_{uw}$ are at least 0.5.
        Therefore we can get the following lower bound on the sum of \Lp{} terms:
		\begin{align*}
			\sum_{u} \lp{u}{v}{w} & \ge (1 - x_{uv}) \cdot \incl{+}{v}{w} + x_{uw} \cdot (\incl{+}{v}{w} + \incl{-}{u}{v} \cdot (1 - \incl{+}{v}{w})) + x_{vw} \cdot \incl{-}{u}{v} \\
			& \geq \frac{1}{2} \incl{+}{v}{w} + \frac{1}{2} (\incl{+}{v}{w} + \incl{-}{u}{v} \cdot (1 - \incl{+}{v}{w})) + x_{vw} \cdot \incl{-}{u}{v}.
		\end{align*} 
		Therefore it suffices to prove that the sum of \Cost{} terms is not greater than this lower bound multiplied by $\alpha$. Rearranging terms, this inequality is equivalent to
		$$
		(\alpha - 1) \cdot  \incl{+}{v}{w} + \alpha x_{vw} \cdot \incl{-}{u}{v} \geq (2 - \frac{\alpha}{2}) \cdot (1 - \incl{+}{v}{w}) \cdot \incl{-}{u}{v}. 
		$$
		Since $\alpha > 1$, it suffices to prove 
		$$
		\alpha x_{vw}  \cdot \incl{-}{u}{v} \geq (2 - \frac{\alpha}{2}) \cdot (1 - \incl{+}{v}{w}) \cdot \incl{-}{u}{v}. 
		$$
		If $\incl{-}{u}{v} = 0$, the inequality follows, otherwise we can cancel this term from both sides and get to the following equivalent inequality
		$$
		\alpha x_{vw} \geq (2 - \frac{\alpha}{2}) \cdot (1 - \incl{+}{v}{w}).
		$$
		By \Cref{lemma:func-ineq}, we know that $1 - \incl{+}{v}{w} \leq \frac{3}{2} x_{vw}$, so it suffices to prove
		$$
		\frac{\alpha}{2 - \frac{\alpha}{2}} \geq \frac{3}{2},
		$$
		and this is equivalent to $\alpha \geq \frac{12}{7}$. Thus in this case, the inequality holds for $\alpha = 2$.
		\item For case (ii), by definition $\incl{-}{u}{v} = 0$, which simplifies the sum of \Cost{} and \Lp{} terms:
		\begin{align*}
			\sum_{u} \cost{u}{v}{w} & = \incl{+}{u}{w} \cdot \incl{+}{v}{w}+  \incl{+}{v}{w} + \incl{+}{u}{w}, \\
			\sum_{u} \lp{u}{v}{w} & = (1 - x_{uv}) \Bigl(1 - (1 - \incl{+}{u}{w}) (1 - \incl{+}{v}{w})\Bigr) + x_{uw} \cdot \incl{+}{v}{w} + x_{vw} \cdot \incl{+}{u}{w}.
		\end{align*}
		We must prove that the sum of \Cost{} terms is not greater than $\alpha$ times the sum of \Lp{} terms, which is equivalent to rearranging the terms as:
		\begin{equation} \tag{$\ddagger$}
			\label{ineq:ppn-0}
			\Bigl(\alpha(1 - x_{uv} + x_{uw}) - 1\Bigr) \incl{+}{v}{w}
			+ \Bigl(\alpha(1 - x_{uv} + x_{vw}) - 1\Bigr) \incl{+}{u}{w}
			\ \ge\
			\Bigl(\alpha(1 - x_{uv}) + 1\Bigr) \incl{+}{u}{w}\incl{+}{v}{w}.
		\end{equation}
        When $x_{uv}$ increases by an amount $\epsilon>0$, the left-hand side decreases by $\alpha(\incl{+}{v}{w}+\incl{+}{u}{w})\epsilon$, while the right-hand side decreases by $\alpha \incl{+}{v}{w} \incl{+}{u}{w}\epsilon < \alpha(\incl{+}{v}{w}+\incl{+}{u}{w})\epsilon$.
        Therefore it suffices to prove the inequality for the largest feasible value $x_{uv} = x_{uw} + x_{vw}$, obtaining
		\begin{equation} \tag{$\star$}
			\label{ineq:ppn}
			\Bigl(\alpha(1 - x_{vw}) - 1\Bigr) \incl{+}{v}{w}
			+ \Bigl(\alpha(1 - x_{uw}) - 1\Bigr) \incl{+}{u}{w} 
			\ \ge\
			\Bigl(\alpha(1 - x_{uw} - x_{vw}) + 1\Bigr) \incl{+}{u}{w}\incl{+}{v}{w}.
		\end{equation}
		To prove Inequality~\ref{ineq:ppn}, we first prove that the function $g_\alpha(x) = (\alpha (1 - x) - 1) \cdot f^+(x)$ is convex.
        \begin{claim}
            $g_\alpha(x)$ is convex on $[0, 0.5)$.
        \end{claim}
        \begin{proof}
            Using \Cref{lemma:der-cont} and the product rule, the derivative of $g$ exists and is:
    		$$
    		g_\alpha'(x) = -\alpha f^+(x) + \bigl(\alpha(1 - x) - 1\bigr) f'^+(x),
    		$$
    		By evaluating this on the sub-domains $[0, 0.25)$ and $[0.25, 0.5)$, we obtain:
    		\begin{align*}
    			g_\alpha'(x)
    			&= \frac{30\alpha x - 19\alpha + 10}{3\sqrt{9 - 20x}} && (\forall x\in[0,0.25)),\\
    			g_\alpha'(x)
    			&= \frac{10}{3}\alpha x - \frac{33}{12}\alpha + \frac{5}{3} && (\forall x\in[0.25,0.5)).
    		\end{align*}
            First, notice that $g_{\alpha}'$ is continuous on $[0.25,0.5)$ as $g'_{\alpha}(0.25^-)=g'_{\alpha}(0.25)=-\frac{23}{12}\alpha + \frac{5}{3}$.
            Furthermore, on the domain $[0, 0.25)$, the numerator of $g'_{\alpha}(x)$ is increasing and the positive denominator is decreasing, and on the domain $[0.25, 0.5)$, the expression is linear and increasing; this implies $g'_{\alpha}$ is increasing.
            Therefore $g_{\alpha}$ is convex on $[0, 0.5)$.
        \end{proof}

       % \nv{Note that by \Cref{lemma:f-product}, we have $\ln(f(x_{uw})) + \ln(f(x_{vw})) \leq 2 \ln(f^+(\frac{x_{uw}+x_{vw}}{2}))$. Therefore,
        %$$
        %f^+(x_{uw})f^+() = 
        %$$
        %}

        We first consider the case that $x_{vw}, x_{uw} \in 0.25 \pm \eta$.
		The convexity of $g(x)$ along with \Cref{lemma:f-product}, imply that setting $ x_{uw} = x_{vw}$ is the tightest case for Inequality~\ref{ineq:ppn}. 
        
        %replacing $(x_{uw}, x_{vw})$ by their average weakens the left-hand side and strengthens the right-hand side.
        Thus, without loss of generality, we can take $x_{uw} = x_{vw} = x\geq 0.25$. The inequality becomes:
		$$
		2\Bigl(\alpha(1 - x) - 1\Bigr) f^+(x) \geq \Bigl(\alpha(1 - 2x) + 1\Bigr) f^+(x)^2,
		$$
		which, after substituting the definition of $f^+(x)$ for $x \in [0.25, 0.5)$, reduces to showing:
		$$
        \frac{2\alpha - 2\alpha x - 2}{\alpha - 2\alpha x + 1} \geq \frac{13}{12} - \frac{5}{3}x.
        $$
    
		The function on the left-hand side is concave for $x\in[0.25,0.5)$, while the right-hand side is linear, therefore it suffices to consider the inequality at endpoints. If $x_{uv} \in 0.5 \pm \slackBad{}$, the endpoints we must consider the inequality at are $x = 0.25$ and $x = 0.5$. For these points, the inequality reduce to $\alpha \geq \frac{16}{7}$ and $\alpha \geq \frac{9}{4}$, and therefore the inequality holds for $\alpha = \frac{16}{7}$. If $x_{uv}\notin 0.5 \pm \slackBad{}$, since $x_{uv} \geq 0.5$, we must have $2x = x_{uv} \geq 0.5 + \slackBad{}$, and therefore $x \geq 0.25 + \frac{\slackBad{}}{2}$, so the endpoints we must consider in this case are $x = 0.25 + \frac{\slackBad{}}{2}$ and $x = 0.5$. At $x = 0.25 + \frac{\slackBad{}}{2}$, expanding gives the following equivalent inequality:
        $$
        \alpha \geq \frac{\frac{8}{3} - \frac{5}{6}\slackBad{}}{\frac{7}{6} + \frac{1}{12} \slackBad{} - \frac{5}{6} \slackBad{}^2} \approx 2.282.
        $$
        Thus, it only remains to verify 
        $$
        \frac{16}{7} - 45 \cdot \charge{} - \gamma \geq \max\left(2.282, \frac{9}{4}\right),
        $$
        which holds.
        
        % The only remaining cases are $x_{uw} \notin 0.25 \pm \slackBad{}$, and $x_{vw} \notin 0.25 \pm \slackBad{}$, without loss of generality, assume $x_{vw} \geq x_{uw}$. Note that if $x_{uw} + x_{vw} \geq 0.5 + \slackBad{}$, the same argument also applies here, so we may assume $x_{uw} + x_{vw} < 0.5 + \slackBad{}$. If $x_{vw} \in 0.25 \pm \slackBad{}$, since $x_{vw} \geq 0.25$, we must have $0.25 \leq x_{vw} < 0.25 + \slackBad{}$, but this along with $0.5 \leq x_{uw} + x_{vw} < 0.5 + \slackBad{}$, imply $0.25 - \slackBad{} \leq x_{uw} \leq 0.25 + \slackBad{}$, which is a contradiction. Thus, we assume $x_{vw} \geq 0.25 + \slackBad{}$, which implies $x_{uw} \leq 0.25$. 
        
       The remaining cases are when $x_{uw} \notin 0.25 \pm \slackBad{}$ or $x_{vw} \notin 0.25 \pm \slackBad{}$.
       If $x_{uw} + x_{vw} \geq 0.5 + \slackBad{}$, we can apply the previous argument as is, where we assumed $x_{vw} = x_{uw}$.
       The only remaining case is when $0.5 \leq x_{uv}= x_{uw} + x_{vw} < 0.5 + \slackBad{}$. 
       For this case, we no longer assume that $x_{uw} = x_{vw}$. Without loss of generality, let $x_{vw} \geq x_{uw}$. Then if $x_{vw} \in 0.25 \pm \slackBad{}$, since $x_{vw} \geq 0.25$, we have $0.25 \le x_{vw} < 0.25 + \slackBad{}$. Combining this with $0.5 \le x_{uw} + x_{vw} < 0.5 + \slackBad{}$, yields $0.25 - \slackBad{} \le x_{uw} \le 0.25 + \slackBad{}$, which contradicts the case assumption.  Hence, we must have $x_{vw} \geq 0.25 + \slackBad{}$, which implies $x_{uw} \le 0.25$.
        
        Denote $x_{uw} + x_{vw}$ and $0.25 + \slackBad{}$ by $s$ and $c$ respectively. We now prove that replacing $(x_{uw}, x_{vw})$ in the Inequality~\ref{ineq:ppn} by $(s -c, c)$ strengthens the inequality. Since $x_{uw} \leq s - c \leq c \leq x_{vw}$ and $x_{uw} + x_{vw} = s$, we can write $c$ and $s-c$ as a convex combination of $x_{uw}$ and $x_{vw}$. Thus, there exists $r\in [0, 1]$ for which
        $$
        c = r\cdot x_{vw} + (1 - r)\cdot x_{uw}, \qquad s - c = (1 - r)\cdot x_{vw} + r\cdot x_{uw}.
        $$
        By \Cref{lemma:f-product}, $\ln (f^+(x))$ is concave, therefore
        \begin{align*}
            \ln(f^+(c)) & = \ln\Bigl(f^+\bigl(r\cdot x_{vw} + (1-r) \cdot x_{uw}\bigr)\Bigr) \geq r\cdot \ln(f^+(x_{vw})) + (1 - r) \cdot \ln(f^+(x_{uw})), \\
            \ln(f^+(s - c)) & = \ln\Bigl(f^+\bigl((1 - r)\cdot x_{vw} + r \cdot x_{uw}\bigr)\Bigr) \geq (1 - r)\cdot \ln(f^+(x_{vw})) + r \cdot \ln(f^+(x_{uw})).
        \end{align*}
        Summing the two inequalities above, we obtain
        $$
        \ln (f^+(c)) + \ln (f^+(s-c)) \geq \ln (f^+(x_{vw})) + \ln (f^+(x_{uw})).
        $$
        Consequently,
        $$
        f^+(c)\cdot f^+(s-c) \geq f^+(x_{vw}) \cdot f^+(x_{uw}).
        $$
        Therefore, replacing $(x_{vw}, x_{uw})$ by $(c, s-c)$ strengthens the right-hand side of Inequality~\ref{ineq:ppn}. Since $g_\alpha(x)$ is convex, we can similarly show that this replacement weakens the left-hand side of Inequality~\ref{ineq:ppn}. In total, this replacement strengthens the inequality and we may assume $x_{vw} = c$ and $x_{uw} = s - c$. The resulting inequality is:
        \begin{equation*} 
			\Bigl(\alpha(1 - c) - 1\Bigr) f^+(c)
			+ \Bigl(\alpha(1 - s + c) - 1\Bigr) f^+(s-c) 
			\geq
			\Bigl(\alpha(1 - s) + 1\Bigr) f^+(c) f^+(s-c).
		\end{equation*}
        Denoting $s - c$ by $x$, substituting $c = 0.25 + \slackBad{}$, and dividing both sides by $f^+(c)f^+(s-c)$ gives
        $$
        T(x)\stackrel{\mathrm{def}}{=}\frac{\alpha(0.75 - \slackBad{}) - 1}{f^+(x)} + \frac{\alpha (1 - x) - 1}{f^+(0.25 + \slackBad{})} - \alpha (0.75 - \slackBad{} - x) - 1 \geq 0.
        $$
        The first term is convex while the other terms are linear, and therefore $T(x)$ is a convex function. 
        Since $0.5 \leq (0.25 + \slackBad{}) + x \leq 0.5 + \slackBad{}$, we have $0.25 - \slackBad{} \leq x \leq 0.25$ and we must prove the inequality on $[0.25 - \slackBad{}, 0.25]$. 
        Since $T(x)$ is convex, it suffices to consider this inequality at the endpoints and also the roots of $T'(x)$ in this interval. First, let us show that $T'(x)$ has no root on this interval.
        \begin{claim}
            $T'(x)$ has no root on $[0.25 - \slackBad{}, 0.25]$.
        \end{claim}
        \begin{proof}
        By expanding, we have
        $$
        T'(x) = \left(\frac{\alpha(0.75-\slackBad{})-1}{f^+(x)}\right)' - \frac{\alpha}{f^+(0.25 + \slackBad{})} + \alpha= \Bigl(\alpha(0.75-\slackBad{})-1\Bigr)\cdot \frac{-f'^+(x)}{f^+(x)^2} - \alpha (\frac{1}{f^+(0.25 + \slackBad{})} - 1).
        $$
        Therefore, since $-f'^+(x)$ is increasing and $f(x)$ is decreasing, by substituting $\slackBad{} = 0.005$, we get
        $$
        \frac{-f'^+(x)}{f(x)^2} \geq \frac{-f'^+(0.25 - \slackBad{})}{f^+(0.25 - \slackBad{})^2} = \frac{10}{9\, f^+(0.25 - \slackBad{})^3} = \frac{10}{9\, f^+(0.245)^3} > 3.61.
        $$
        Consequently, 
        $$
        T'(x) \geq \Bigl(\alpha \cdot 0.745 - 1\Bigr) \cdot 3.61 - \alpha (\frac{1}{f^+(0.255)} - 1) > 2.68 \alpha - 3.61 - \alpha \cdot 0.6 = 2.08 \alpha - 3.61 > 0, 
        $$
        which holds, since $\alpha \geq 2$.
        \end{proof}
        Thus, it suffices to consider the inequality at endpoints $x= 0.25 - \slackBad{}$, and $x = 0.25$. At $x = 0.25$, by rearranging the terms, the inequality simplifies to
        $$
        \alpha \geq \frac{\frac{5}{2} + \frac{1}{f^+(0.25+\slackBad{})}}{\frac{5}{8} - \frac{1}{2} \slackBad{} + \frac{3}{4\,f^+(0.25+\slackBad{})}} = \frac{\frac{5}{2} + \frac{1}{f^+(0.255)}}{\frac{5}{8} - \frac{0.005}{2}  + \frac{3}{4\,f^+(0.255)}} \approx 2.282.
        $$
        At $x = 0.25 - \slackBad{}$, by rearranging the terms, the inequality simplifies to
        $$
        \alpha \geq \frac{\frac{1}{f^+(0.25 + \slackBad{})} + \frac{1}{f^+(0.25 - \slackBad{})} + 1}{\frac{0.75 - \slackBad{}}{f^+(0.25 - \slackBad{})} + \frac{0.75 + \slackBad{}}{f^+(0.25 + \slackBad{})} - \frac{1}{2}} = \frac{\frac{1}{f^+(0.255)} + \frac{1}{f^+(0.245)} + 1}{\frac{0.74}{f^+(0.245)} + \frac{0.76}{f^+(0.265)} - \frac{1}{2}} \approx 2.2852.
        $$
        Hence, it remains to verify
        $$
        \frac{16}{7} - \gamma \geq \max \left(2.282, 2.2852\right) = 2.2852,
        $$
        which clearly holds.
		\item For final case (iii), assume $x_{uw}$ and $x_{vw}$ are  fixed. Therefore, by rearranging, the required inequality becomes
		$$
		l(x_{uv}) \stackrel{\mathrm{def}}{=} A x_{uv} + B \incl{-}{u}{v} + C \geq 0,
		$$
		where $A$, $B$ and $C$ are defined as follows:
		\begin{align*}
			A & = \alpha \cdot \incl{+}{u}{w} \cdot \incl{+}{v}{w} -\alpha \cdot \incl{+}{u}{w} - \alpha \cdot \incl{+}{v}{w}, \\
			B & = \alpha \cdot (x_{uw} + x_{vw}) + 2 \incl{+}{v}{w} + 2 \incl{+}{u}{w} - \alpha x_{uw} \incl{+}{v}{w} - \alpha x_{vw} \incl{+}{u}{w} - 2, \\
			C & = \Bigl(\alpha (1 + x_{vw}) - 1\Bigr) \incl{+}{u}{w} + \Bigl(\alpha (1 + x_{uw}) - 1\Bigr) \incl{+}{v}{w} - (\alpha + 1) \cdot \incl{+}{u}{w} \cdot \incl{+}{v}{w}.
		\end{align*}
		The derivative of the function $l(x_{uv})$ is:
		\begin{align*}
			l'(x_{uv}) = A + B\cdot \frac{df^-}{dx}(x_{uv}) = A - B.
		\end{align*}
		Since $ \incl{+}{v}{w} \incl{+}{u}{w} < \incl{+}{v}{w}+\incl{+}{u}{w}$, the expression $A$ is non-positive. Hence, to show that $l'(x_{uv}) \leq 0$, it is sufficient to prove that $B \geq 0$. This can be shown by combining the following two inequalities, which hold by \Cref{lemma:func-ineq} and $x_{uw}, x_{vw} < 0.5$. For any $p \in [0, 1]$:
		\begin{align*}
			\Bigl(\frac{3}{2}x_{uw}+\incl{+}{u}{w}\Bigr)p + \Bigl(\frac{3}{2}x_{vw}+\incl{+}{v}{w}\Bigr)p & \geq 2p,\\
			\Bigl(\alpha - \frac{3}{2}p\Bigr)(x_{uw}+x_{vw}) + \frac{3}{4}p(\incl{+}{u}{w}+\incl{+}{v}{w}) & \geq \alpha x_{uw}\incl{+}{v}{w} + \alpha x_{vw}\incl{+}{u}{w}.
		\end{align*}

		Combining these gives the bound:
		\begin{equation}
			\tag{$\dagger$}
			\label{ineq:A}
			B \geq (2 - \frac{7}{4} p) (\incl{+}{u}{w} + \incl{+}{v}{w}) - (2 - 2p).
		\end{equation}
		Setting $p = 1$ yields
		$$
		B \geq \frac{1}{4} (\incl{+}{u}{w} + \incl{+}{v}{w}) \geq 0.
		$$
		Therefore, $B \geq 0$ and consequently, $l(x_{uv})$ is a decreasing function of $x_{uv}$ and we may take $x_{uv} = \min(x_{uw} + x_{vw}, 0.5^-)$. If $x_{uw} + x_{vw} \geq 0.5$, since $B$ and $\incl{-}{u}{v}$ are both non-negative, it suffices to prove  
		$$
		A x_{uv} + C \geq 0.
		$$
		By rearranging, this transforms to
		$$
		\Bigl(\alpha (1 + x_{vw} - x_{uv}) - 1\Bigr) \incl{+}{u}{w} + \Bigl(\alpha (1 + x_{uw} - x_{uw}) - 1\Bigr) \incl{+}{v}{w} \geq \Bigl(\alpha (1 - x_{uv}) + 1\Bigr) \incl{+}{u}{w} \incl{+}{v}{w}.
		$$
		This inequality 
        is identical to Inequality~\ref{ineq:ppn-0}, which has been proved in Case (ii).
        
        Now, let us consider the case that $x_{uw} + x_{vw} < 0.5$, and it follows from the above discussion that $x_{uv} = x_{uw} + x_{vw}$. Now set $p = \frac{4}{7}$ in Inequality~\ref{ineq:A} to obtain
		$$
		B \geq (\incl{+}{v}{w} + \incl{+}{u}{w}) - \frac{6}{7}.
		$$
		With this, it now suffices to show
		$$
		A x_{uv} + \Bigl((\incl{+}{v}{w} + \incl{+}{u}{w}) - \frac{6}{7}\Bigr) \cdot \incl{-}{u}{v} + C \geq 0.
		$$
		Let $c' = x_{uv} = x_{uw} + x_{vw}$ be fixed and $c = f^-(c')$. As in the earlier convexity argument, one can similarly show that 
        the function $g_\alpha(x) = \Bigl(\alpha (1 - x) - 1 + c\Bigr) f^+(x)$ is convex. By rearranging the above inequality, we can obtain the equivalent inequality
		$$
		g_\alpha(x_{uw}) + g_\alpha(x_{vw}) \geq \Bigl(\alpha (1 - c') + 1\Bigr) \incl{+}{u}{w} \incl{+}{v}{w} + \frac{6}{7} c.
		$$
		Therefore, by \Cref{lemma:f-product}, 
        the setting $x_{uw} = x_{vw} = x < 0.25$ is the tightest case for the inequality above. 
     Then, $c' = 2x$ and $c = f^-(2x) = 1 - 2x$. So, we must show
		$$
		2\Bigl(\alpha(1-x) - 1 + (1-2x)\Bigr) f^+(x)
		\geq
		\Bigl(\alpha(1-2x) + 1\Bigr) f^+(x)^2 + \frac{6}{7}(1 - 2x).
		$$
		Let us define the function $\Phi(x)$ as the left-hand side minus the right-hand side:
		$$
		\Phi(x) \stackrel{\mathrm{def}}{=} 2\Bigl(\alpha(1-x) - 1 + (1-2x)\Bigr) f^+(x)
		-
		\Bigl(\alpha(1-2x) + 1\Bigr) f^+(x)^2 - \frac{6}{7}(1 - 2x).
		$$
		We analyze $\Phi(x)$ for $x\in [0, 0.25)$ by performing a change of variables. Let
		$$
		u:=f^+(x)=\sqrt{\,1-\frac{20}{9}x\,}\in\Bigl(\frac{2}{3},1\Bigr].
		$$
        Then, $x=\frac{9}{20}(1-u^2),
		1-2x=\frac{1+9u^2}{10}.$
		With these substitutions and by letting $\alpha = 2$, we compute the necessary coefficients:
		\begin{align*}
			& \alpha(1-x) - 1 + (1-2x) = \frac{9u^2 + 1}{5}, \\
			& \alpha(1-2x) + 1 = \frac{9u^2 + 6}{5}.
		\end{align*}
		Substituting, we get
		$$
		35\,\Phi(x)=Q(u):=-63u^4+126u^3-69u^2+14u-3.
		$$
		Since $u \in (\frac{2}{3}, 1]$, we have $(u-\frac{2}{3})(u-1)\leq 0$. Expanding gives $u^2 + \frac{2}{3} \leq \frac{5}{3} u$, which is equivalent to
        $$
        u - \frac{3}{5} u^2 - \frac{2}{5} \geq 0.
        $$
        By rearranging, we get
        $$
        Q(u)= 126u^2 \cdot (u - \frac{3}{5} u^2 - \frac{2}{5}) + 14 \cdot (u - \frac{3}{5} u^2 - \frac{2}{5}) + (\frac{63}{5} u^4 - \frac{51}{5} u^2 + \frac{13}{5}).
        $$
        The first two terms above are clearly nonnegative, therefore for proving $Q(u) \geq 0$, it suffices to prove the last term is also nonnegative. Denoting $u^2$ by $r$ leaves us with
        $$
        T(r) = 63r^2 - 51r + 13 \geq 0.
        $$
        $T(r)$ is a quadratic function and therefore attains its minimum at $\frac{51}{126}$. Since $T(\frac{51}{126}) = 13 - \frac{51^2}{2 \cdot 126} > 0$, $T$ is positive on its entire domain.
        Hence $Q(u) \geq 0$ on the interval $(\frac{2}{3}, 1]$, consequently $\Phi(x) \geq 0$ on the interval $[0, 0.25)$ for $\alpha \geq 2$. In this case we have $\frac{\Cost{T}}{\Lp{T}} \leq 2$, so it remains to verify
        $$
        \frac{16}{7} - 45\cdot \charge{} - \gamma \geq 2,
        $$
        which holds. \qedhere
	\end{itemize}
\end{proof}
\subsubsection{\texorpdfstring{$+++$ Triangles}{+++ Triangles}}
\begin{lemma}\label{lemma:+++_triangle}
	For a $+++$ triangle $T=uvw$,
	$$
	\frac{\Cost{T}}{\Lp{T}} \leq \frac{16}{7}.
	$$

    Furthermore:
    \begin{itemize}
        \item if $x_{uw},x_{vw}\in 0.25\pm \slackBad{}$ and $x_{uv} \in 0.5\pm \slackBad{}$, then $\frac{\Cost{T}-\charge{}}{\Lp{T}} \le \frac{16}{7} - \gamma$,
        \item if $x_{uw},x_{vw}\in 0.25\pm \slackBad{}$ and $x_{uv} \not\in 0.5\pm \slackBad{}$, then $\frac{\Cost{T} + 9\charge{}}{\Lp{T}} \le \frac{16}{7} - \gamma$,
        \item in all other cases $\frac{\Cost{T}}{\Lp{T}} \le \frac{16}{7} - \gamma$.
    \end{itemize}
\end{lemma}

\begin{proof}
	We begin by computing the \Cost{} and \Lp{} terms for each of the three pivots.
	\begin{align*}
		\cost{u}{v}{w} & = (1 - \incl{+}{u}{v}) \cdot \incl{+}{u}{w} + \incl{+}{u}{v} \cdot (1 - \incl{+}{u}{w}), \\
		\lp{u}{v}{w} & =  x_{vw} \cdot \Bigl(1 - (1 - \incl{+}{u}{v}) (1 - \incl{+}{u}{w})\Bigr), \\
		\cost{v}{u}{w} & = (1 - \incl{+}{u}{v}) \cdot \incl{+}{v}{w} + \incl{+}{u}{v} \cdot (1 - \incl{+}{v}{w}), \\
		\lp{v}{u}{w} & = x_{uw} \cdot \Bigl(1 - (1 - \incl{+}{u}{v}) (1 - \incl{+}{v}{w})\Bigr), \\
		\cost{w}{u}{v} & = (1 - \incl{+}{u}{w}) \cdot \incl{+}{v}{w} + \incl{+}{u}{w} \cdot (1 - \incl{+}{v}{w}), \\
		\lp{w}{u}{v} & = x_{uv} \cdot \Bigl(1 - (1 - \incl{+}{u}{w}) (1 - \incl{+}{v}{w})\Bigr).
	\end{align*}
Note that since $\Lp{T}$ is sum of 3 terms, each being at most one, it is clearly upper bounded by three.
    Therefore, since $\gamma \leq 3 \cdot \charge{}$, for $x_{uw}, x_{vw} \in 0.25 \pm \slackBad{}$ and $x_{uv} \in 0.5 \pm \slackBad{}$,
    $$
    \frac{\Cost{T} - \charge{}}{\Lp{T}} = \frac{\Cost{T}}{\Lp{T}} - \frac{\charge{}}{\Lp{T}} \leq \frac{\Cost{T}}{\Lp{T}} - \frac{\charge{}}{3} \leq \frac{\Cost{T}}{\Lp{T}} - \gamma.
    $$
    So in this case, it remains to prove that $\frac{\Cost{T}}{\Lp{T}} \leq \frac{16}{7}$. Also if $x_{uw}, x_{vw}\in 0.25 \pm \slackBad{}, x_{uv}\not\in 0.5\pm \slackBad{}$, by \Cref{lem:enoughBudget}:
    $$
    \frac{\Cost{T} + 9\charge{}}{\Lp{T}} = \frac{\Cost{T}}{\Lp{T}} + \frac{9}{0.2}\charge{} \leq \frac{\Cost{T}}{\Lp{T}} + 45 \cdot \charge{}.
    $$
    So in this case it remains to prove $\frac{\Cost{T}}{\Lp{T}} \leq \frac{16}{7} - 45 \cdot \charge{} - \gamma$. In all other cases, we must prove $\frac{\Cost{T}}{\Lp{T}} \leq \frac{16}{7} - \gamma$, thus in each case we must prove that $\frac{\Cost{T}}{\Lp{T}}$ is upper bounded by some constant.
    
Let $\alpha \in [2,\frac{16}{7}]$ be a parameter. We show that:
$$
\sum_{u} \cost{u}{v}{w} \leq \alpha \cdot \sum_{u} \lp{u}{v}{w},
$$
where the sum is cyclic over $uvw$. We now divide into cases: (i) at least two of $x_{uv}$, $x_{uw}$ and $x_{vw}$ are at least 0.5, (ii) exactly one of these values is at least 0.5, and (iii) all three are less than 0.5.
\begin{itemize}
	\item For case (i), without loss of generality, assume $x_{uv}, x_{uw} \geq 0.5$. This implies $f^+(x_{uv}) = f^+(x_{uw})=0$, and hence we can simplify the sum of \Cost{} and \Lp{} terms into
	\begin{align*}
		\sum_{u} \cost{u}{v}{w} & = 2 \incl{+}{v}{w}, \\
		\sum_{u} \lp{u}{v}{w} & = (x_{uv} + x_{uw}) \cdot \incl{+}{v}{w}.
	\end{align*}
	Since $x_{uv}$ and $x_{uw}$ are both at least 0.5, it follows that $\sum_{u} \cost{u}{v}{w} \leq 2 \cdot \sum_{u} \lp{u}{v}{w}$.
	\item For case (ii), without loss of generality, assume $x_{uv} \geq 0.5$, while $x_{uw} \le x_{vw} < 0.5$. The sum of \Cost{} and \Lp{} terms simplifies to
	\begin{align*}
		\sum_{u} \cost{u}{v}{w} & = 2 \incl{+}{u}{w} + 2 \incl{+}{v}{w} - 2 \incl{+}{u}{w} \incl{+}{v}{w}, \\
		\sum_{u} \lp{u}{v}{w} & = x_{uv} \cdot \Bigl( \incl{+}{u}{w} + \incl{+}{v}{w} - \incl{+}{u}{w} \incl{+}{v}{w} \Bigr) + x_{uw} \cdot \incl{+}{v}{w} + x_{vw} \cdot \incl{+}{u}{w}.
	\end{align*}
	Since the sum of \Cost{} terms is independent of $x_{uv}$, while the sum of \Lp{} terms is increasing in $x_{uv}$, it suffices to prove the claim for the minimum possible value of $x_{uv}$, that is $\max\set{0.5,|x_{uw}-x_{vw}|}$ ($x_{uv} \ge 0.5$ by assumption, and $x_{uv} \ge |x_{uw}-x_{vw}|$ by triangle inequality).
    As both $x_{uw},x_{vw} < 0.5$, we have $|x_{uw}-x_{vw}|<0.5$, and thus it suffices to prove the claim for $x_{uv}=0.5$. By rearranging, it remains to prove
	$$
	\alpha(x_{uw} \cdot \incl{+}{v}{w} + x_{vw} \cdot \incl{+}{u}{w}) \geq (2 - \frac{\alpha}{2}) \Bigl(\incl{+}{u}{w} + \incl{+}{v}{w} - \incl{+}{u}{w} \incl{+}{v}{w}\Bigr).
	$$
	Let $s = x_{uw} + x_{vw}, f=f^+$.
    Notice that $s\in [0.5,1]$ as $0.5=x_{uv}\le x_{uw} + x_{vw} < 0.5 + 0.5 = 1$.
    Define
	$$
	\Phi_{s,\alpha}(x) = \alpha\Bigl(x f(s-x) + (s-x)f(x)\Bigr)
	- (2 - \frac{\alpha}{2})\Bigl(f(x)+f(s-x) - f(x)f(s-x)\Bigr).
	$$
    We now show that $\Phi_{s,\alpha}(x)$ is convex.
    As $\Phi_{s,\alpha}(x) = \Phi_{s,\alpha}(s-x)$, the convexity proves that the minimum of $\Phi_{s,\alpha}$ is attained at $x=s/2$, which means that it suffices to show that $\Phi_{s,\alpha}(s/2) \ge 0$.

    \begin{claim}
    $\Phi_{s,\alpha}(x)$ is convex in $[x_{uw}, s]$.
    \end{claim}
    \begin{proof}
    By \Cref{lemma:der-cont}, $f$ is continuous and has a continuous first derivative on $[0,0.5)$.
    Hence $\Phi_{s,\alpha}$ is also continuous and has a continuous first derivative on $(s-\frac12,\frac12)$, that is the interval where both $x$ and $s-x$ are in $[0,0.5)$.
	\begin{align*}
		\Phi_{s,\alpha}'(x)
		&= \alpha \Bigl(f(s-x) - x f'(s-x) - f(x) + (s-x)f'(x)\Bigr) \\
		&\quad - (2 - \frac{\alpha}{2})\Bigl(f'(x) - f'(s-x) - f'(x)f(s-x) + f(x)f'(s-x)\Bigr),
	\end{align*}
	and where $f''$ exists,
	\begin{align*}
		\Phi_{s,\alpha}''(x)
		&= \alpha\Bigl(-2 f'(s-x) + x f''(s-x) - 2 f'(x) + (s-x)f''(x)\Bigr) \\
		&\quad - (2 - \frac{\alpha}{2})\Bigl(f''(x)+f''(s-x) - f''(x)f(s-x) - f(x)f''(s-x) + 2 f'(x)f'(s-x)\Bigr).
	\end{align*}
    Note that by definition of $f^+$, we have that $f''$ exists for every $x\ne 0.25$; but if $0.25 \in [x_{uw},s]$, then $s - 0.5 < s - x_{vw} = x_{uw} \le 0.25 < 0.5$, therefore $0.25$ is included in the interval where $\Phi_{s,\alpha}$ has a continuous first derivative.
    We conclude that to prove convexity of $\Phi_{s,\alpha}$, it suffices to prove $\Phi_{s,\alpha}''>0$ where it is defined.
	
	When both $x,s-x\ge \frac14$ we have $f(x) = f^+(x) = \frac{13}{12}-\frac{5}{3}x$, and thus, $f''\equiv 0$ and $f'\equiv -\frac{5}{3}$, hence
	$$
	\Phi_{s,\alpha}''(x)=\alpha\cdot 4\cdot\frac{5}{3} - (2 - \frac{\alpha}{2})\cdot2\cdot\Bigl(\frac{5}{3}\Bigr)^2 > 0,
	$$
    because $\alpha \ge 2$.
        
	When $x<\frac14\le s-x$, writing $u:=f(x)\in(\frac{2}{3},1]$ and $L:=f(s-x)$, using
	$$
	f'(x)=-\frac{10}{9}\frac{1}{u},\quad f''(x)=-\frac{100}{81}\frac{1}{u^3},\quad
	f'(s-x)=-\frac{5}{3},\quad f''(s-x)=0,
	$$
	one obtains
	$$
		\Phi_{s,\alpha}''(x)
		= \alpha\Bigl(\frac{10}{3} + \frac{20}{9} \frac{1}{u} - (s - x) \frac{100}{81} \frac{1}{u^3}\Bigr) 
        - 
        (2 - \frac{\alpha}{2})\Bigl(\frac{100}{81} \frac{L}{u^3} -\frac{100}{81} \frac{1}{u^3} + \frac{100}{27} \frac{1}{u}\Bigr).
	$$
	Since $L=\frac{13}{12} - \frac{5}{3}(s-x)$ with rearranging and multiplying by $243$,
	$$
	243\,\Phi''_{s,\alpha} (x) = 810\alpha + \frac{990 \alpha - 1800}{u} + \frac{1000(s-x)-550\alpha (s-x) - 50 + 12.5 \alpha}{u^3}.
	$$
	Since the  coefficient of $\alpha$ is positive, it suffices to prove convexity for $\alpha = 2$, where $243 \,\Phi''_{s,\alpha}$ simplifies to
	$$
	1620 + \frac{180}{u} - \frac{100 (s - x) + 25}{u^3}\geq 1620 - \frac{100(s-x)+25}{u^3} \geq 1620 - (100 + 25) \cdot \left(\frac{3}{2}\right)^3 > 0.
	$$
	Thus $\Phi_{s,\alpha}''(x)>0$ when $x<\frac{1}{4}\leq s - x$. Therefore $\Phi_{s,\alpha}$ is convex.
    \end{proof}
    
	We conclude that $\Phi_{s,\alpha}(x)$ attains its minimum at $x = \frac{s}{2}$. Thus, it suffices to show that $\Phi_{s,\alpha}(s/2) \ge 0$ for $s\in[0.5,1]$, or equivalently that for $x\in[\frac{1}{4},\frac{1}{2})$,
	$$
	2\alpha x f^+(x) \;\ge\; \Bigl(2 - \frac{\alpha}{2}\Bigr)\bigl(2f^+(x) - f^+(x)^2\bigr),
	$$
	which for $\alpha = \frac{16}{7}$ is equivalent to
	$$
	16x \geq 3\Bigl(2 - \bigl(\frac{13}{12} - \frac{5}{3}x\bigr)\Bigr),
	$$
	i.e., $x\ge \frac{1}{4}$, which holds.
    
    Furthermore, for $x_{uv}, x_{uw} \in 0.25 \pm \slackBad{}$ and $x_{vw} \notin 0.5 \pm \slackBad{}$, since $x_{uv} \geq 0.5$, we have $x_{uv} > 0.5 + \slackBad{}$, which means $s = x_{uw} + x_{vw} > 0.5 + \slackBad{}$. Thus, in this case $\Phi_{s,\alpha}$ attains its minimum at $x = \frac{s}{2} > 0.25 + \frac{\slackBad{}}{2}$. As it suffices to prove
    $$
    \frac{2\alpha}{2 - \frac{\alpha}{2}} \geq \frac{2 - f^+(x)}{x},
    $$
    hence it remains to show
    $$
    \frac{2\alpha}{2 - \frac{\alpha}{2}} \geq \frac{11}{12 (0.25 + \frac{\slackBad{}}{2})} + \frac{5}{3}.
    $$
    By rearranging, this transforms to
    $$
    \alpha \geq \frac{\frac{22}{3+12\slackBad{}} + \frac{10}{3}}{\frac{17}{6} + \frac{11}{6 + 24\slackBad{}}} \approx 2.273.
    $$
    Therefore, it suffices to verify
    $$
    \frac{16}{7} - 45 \cdot \charge{} - \gamma \geq 2.273,
    $$
    which holds.

    The only remaining cases are when one of $x_{uw},x_{vw} \notin 0.25 \pm \slackBad{}$. Without loss of generality, assume $x_{vw} \geq x_{uw}$. Note that if $x_{uw} + x_{vw} > 0.5 + \slackBad{}$, the same argument as before applies, so we may assume $x_{uw} + x_{vw} 
    \le 0.5 + \slackBad{}$. If $x_{vw} \in 0.25 \pm \slackBad{}$, since $x_{vw} \geq 0.25$, we must have $0.25 \leq x_{vw} \le 0.25 + \slackBad{}$, but this along with $0.5 \leq x_{uw} + x_{vw} \le 0.5 + \slackBad{}$, imply $0.25 - \slackBad{} \leq x_{uw} \leq 0.25 + \slackBad{}$, which is a contradiction. Thus, we assume $x_{vw} > 0.25 + \slackBad{}$, which implies $x_{uw} \leq 0.25$. 

    Let $c = 0.25 + \slackBad{}$. Since $x_{vw} \geq 0.25 + \slackBad{} \geq s - 0.25 - \slackBad{} \geq x_{uw}$, there exists $r\in [0, 1]$ such that
    $$
    c = r \cdot x_{vw} + (1 - r) \cdot x_{uw}, \qquad s - c = (1 - r) \cdot x_{vw} + r \cdot x_{uw}.
    $$
    Since $\Phi_{s, \alpha}$ is convex,
    \begin{align*}
    \Phi_{s,\alpha} (c) & = \Phi_{s,\alpha} (r \cdot x_{vw} + (1 - r) \cdot x_{uw}) \leq r \cdot \Phi_{s,\alpha} (x_{vw}) + (1 - r) \cdot \Phi_{s,\alpha} (x_{uw}), \\
    \Phi_{s,\alpha} (s - c) & = \Phi_{s,\alpha} ((1 - r) \cdot x_{vw} + r \cdot x_{uw}) \leq (1 - r) \cdot \Phi_{s,\alpha} (x_{vw}) + r \cdot \Phi_{s,\alpha} (x_{uw}).
    \end{align*}
    Summing the two inequalities above, we obtain
    $$
    \Phi_{s,\alpha} (c) + \Phi_{s,\alpha} (s-c) \leq \Phi_{s,\alpha} (x_{uw}) + \Phi_{s,\alpha} (x_{vw}),
    $$
    as $\Phi_{s,\alpha} (x) = \Phi_{s,\alpha} (s - x)$, this is equivalent to 
    $$
    \Phi_{s,\alpha} (c) \leq \Phi_{s,\alpha} (x_{vw}).
    $$
    We can therefore without loss of generality, assume that $x_{vw} = c$. Denoting $s - c$ by $y$ and substituting $c$, it suffices to prove for $y\in [0.25 - \slackBad{}, 0.25]$,
    $$
    T(y) = \alpha\Bigl( (0.25 + \slackBad{}) f(y) + y f(0.25 + \slackBad{})\Bigr) - (2 - \frac{\alpha}{2}) \Bigl(f(0.25 + \slackBad{}) + f(y) - f(0.25 + \slackBad{}) f(y)\Bigr) \geq 0.
    $$
    \begin{claim}
        $T(y)$ is an increasing fucntion on the domain $[0.25 - \slackBad{}, 0.25]$.
    \end{claim}
    \begin{proof}
    It suffices to prove $T'(y) \geq 0$ on this interval. We have
    $$
    T'(y) = f'(y) \Bigl(\alpha(0.25 + \slackBad{}) - (2 - \frac{\alpha}{2}) (1 - f(0.25 + \slackBad{}))\Bigr) + \alpha f(0.25 + \slackBad{}).
    $$

    Since $f'$ is a negative decreasing function and $\alpha f(0.25 + \slackBad{})$ is positive, if the term $\Bigl(\alpha(0.25 + \slackBad{}) - (2 - \frac{\alpha}{2}) (1 - f(0.25 + \slackBad{}))\Bigr)$ is negative, $T'(y) \geq 0$ immediately follows. Otherwise it suffices to consider $y = 0.25$. For this choice, by substituting $\slackBad{} = 0.005$, since $0.65 < f^+(0.255) < 0.66$
    $$
    T'(0.25) \geq \frac{-5}{3} \Bigl(0.255\alpha - (2 - \frac{\alpha}{2}) (1 - 0.66)\Bigr) + \alpha \cdot 0.65.
    $$
    Thus,
    $$
    T'(0.25) \geq 1.13 - 0.059\alpha > 0,
    $$
    which holds for $\alpha < 19.1$.
    \end{proof}
    So it suffices to consider $y = 0.25 - \slackBad{}$. For this choice, since $\slackBad{} = 0.005$, we have
    $$
    T(y) = \alpha \Bigl(0.255\cdot f(0.245) + 0.245\cdot f(0.255)\Bigr) - (2 - \frac{\alpha}{2})\Bigl(f(0.245) + f(0.255) - f(0.245) f(0.255)\Bigr).$$
    Thus,
    $$
    T(y)\ge 0.3334 \cdot \alpha - 0.889 \cdot (2 - \frac{\alpha}{2}) = 0.7779 \cdot \alpha - 1.778.
    $$
    $T(y)\ge 0$ holds for $\alpha \in [2.28565, \frac{16}{7}]$. Since $\frac{16}{7} - \gamma \geq 2.28565$, the desired result is proved.
	\item For case (iii), the inequality to prove is
	$$
	\alpha \sum_{u} x_{vw}\Bigl(f^+(uv)+f^+(uw) - f^+(uv)f^+(uw)\Bigr)
	\geq
	2\sum_{u} f^+(x_{vw}) - 2\sum_{u} f^+(uv)f^+(uw).
	$$
    where the sums are cyclic over $uvw$.~\Cref{lemma:func-ineq} implies $x \geq \frac{9}{20} (1 - f^+(x)^2)$, therefore it suffices to show
	$$
	\alpha \sum_u \Bigl(\frac{9}{20}(1 - f^+(x_{vw})^2)\Bigr)\Bigl(f^+(uv)+f^+(uw) - f^+(uv)f^+(uw)\Bigr)
	\geq
	2\sum_{u} f^+(x_{vw}) - 2\sum_{u} f^+(uv)f^+(uw).
	$$
	Let $A=f^+(x_{vw})$, $B=f^+(x_{uv})$, $C=f^+(uw)$. For the remaining of this case, sums are cyclic over $ABC$. Expanding, and setting $\beta = \frac{9}{20}\alpha$ we get
	$$
	\beta\Bigl(2\sum A - \sum AB + ABC(A+B+C) - \sum AB(A+B)\Bigr)
	\geq
	2\sum A - 2\sum AB.
	$$
	Let $a=1-A$, $b=1-B$, $c=1-C$, and define $S=a+b+c$, $P=ab+ac+bc$, $R=abc$. Then
    
	\begin{align*}
		\sum A &= 3 - S, \\
		\sum AB &= 3 - 2S + P, \\
		ABC &= 1 - S + P - R, \\
		\sum AB(A+B) &= (\sum AB)(\sum A) - 3ABC = 6 - 6S + 2S^2 - SP + 3R.
	\end{align*}

	Substituting, the inequality becomes
	$$
	(-2+2\beta)S - \beta S^2 + (2+2\beta)P + \beta(S - 6)R \geq 0.
	$$
	Since $A,B,C\in (\frac{1}{4},1]$ (by definition of $f^+$), we have $a,b,c\in[0,\frac{3}{4})$ and thus $S-6<0$. By \Cref{lemma:func-ineq}, $\frac{10}{9}x_{vw} \leq a \leq \frac{3}{2}x_{vw}$, and similarly for $b,c$. Hence
	$$
	\frac{9}{10}a + \frac{9}{10}b \geq x_{vw}+x_{uv} \geq x_{vw} \geq \frac{2}{3}c \quad\Rightarrow\quad c \leq \frac{27}{47}S,
	$$
	and the same bound holds for $a$ and $b$.
	
	Let 
	$$
	\Phi(a,b,c)=(-2+2\beta)S - \beta S^2 + (2+2\beta)P + \beta(S - 6)R
	$$
	and without loss of generality, assume $a\geq b$. Consider $p(t)=\Phi(a+t,b-t,c)$. For any $t > 0$,
	$$
	p'(t) = -(2+2\beta)(a-b+2t) + \beta c(6 - S)((a-b)+2t)
	\leq
	(a-b+2t)\Bigl(\beta c(6 - S) - (2+2\beta) \Bigr).
	$$
	Since $c\leq \frac{3}{4}$ we have $c(6-S)\le c(6-c)\leq \frac{3}{4}\cdot\frac{21}{4}$, hence
	$$
	p'(t) \leq (a-b+2t)\Bigl(\frac{63}{16}\beta - (2+2\beta)\Bigr).
	$$
    which is negative for $\beta < \frac{32}{31}$.

	Thus, for such values of $\beta$, it suffices to prove $\Phi(a,b,c)\ge 0$ under the constraints:
    \begin{itemize}
        \item $a,b,c\in [0,\max\set{0.75,\frac{27}{47}S}]$
        \item $a+b+c = S$
        \item among any two of $a,b,c$, either one equals $\frac{27}{47}S$ or one is $0$.
    \end{itemize} 
    We note that this is only a subset of the actual conditions $a,b,c$ must satisfy; for example, we allow for values of $a,b,c$ that would violate the triangle inequality on $x_{uv},x_{uw},x_{vw}$.
    Therefore, proving $\Phi(a,b,c)\ge 0$ under these conditions is sufficient.
    
    If $S>0$, we cannot have more than one of $a,b,c$ equal to $0$ (because then the other one would be equal to $S$ and therefore larger than $\frac{27}{47}S$), and we cannot have two of $a,b,c$ equal to $\frac{27}{47}S$ because then the sum is $\frac{54}{47}S > S$.
    Therefore, it suffices to consider $(a,b,c)=(\frac{27}{47}S,\frac{20}{47}S,0)$; notice that this covers the case $S=0$. For this choice, we get
	$$
	(-2+2\beta)S + \Bigl((2+2\beta)\cdot \frac{27\cdot 20}{47^2} - \beta\Bigr) S^2 \;\ge\; 0.
	$$
	For $\beta=\frac{37}{36} < \frac{32}{31}$ (therefore $\alpha = \frac{20}{9} \cdot \frac{37}{36} = \frac{185}{81} < \frac{16}{7})$ this is true when $S\in [0,1.527134]$. As $0\le S = a+b+c \le 0.75 + 0.75 + 0 = 1.5$, the result follows.

    To bound $\frac{\Cost{T}+9\charge{}}{\Lp{T}}$ when $x_{uv},x_{uw}\in 0.25\pm \slackBad{}$, we have $\frac{\Cost{T}+9\charge{}}{\Lp{T}} \le \frac{185}{81} + \frac{9\charge}{\Lp{T}}$.

    % We have that $\lp{v}{u}{w} = x_{uw} \cdot \Bigl(1 - (1 - \incl{+}{u}{v}) (1 - \incl{+}{v}{w})\Bigr) \ge x_{uw} \incl{+}{v}{w})\Bigr) \ge (0.25-\slackBad{}) f^+(0.25+\slackBad{})$, by the monotonicity of $f^+$.
    % This is at least $0.1$ for sufficiently small $\slackBad{}$; similarly for $\lp{w}{u}{v}$.
    By \Cref{lem:enoughBudget}, $\Lp{T} \geq 0.2$. Therefore,
    $$\frac{\Cost{T}+9\charge{}}{\Lp{T}} \le \frac{185}{81} + \frac{9\charge}{\Lp{T}} \leq \frac{185}{81} + 45 \cdot \charge{} \leq \frac{16}{7} - \gamma,$$ 
    which holds. \qedhere
\end{itemize}
\end{proof}

\section{\texorpdfstring{Approximation beyond $16/7$}{Approximation beyond 16/7}}\label{sec:beyond}
In this section, we again use \Cref{alg:pivot-supernodes}, but instead of solving the standard LP (\Cref{fig:LP}), we solve the Sherali-Adams relaxation (\Cref{fig:SA}), which also gives us an LP distance $x_{uv}$ for each pair of points.

To go beyond $\frac{16}{7}$, we employ a strategy similar to the one from~\cite{sub2}.
The idea is that there cannot be too many \emph{bad} triangles (i.e., triangles with ratio very close to $\frac{16}{7}$); there will always exist a fraction of triangles with much smaller ratio, and these ``drag'' the approximation ratio below $\frac{16}{7}$.

To formalize this intuition, we define \emph{bad} triangles to be triangles $uvw$ with $x_{uv}\in 0.25\pm \slackBad{}, x_{uw}\in 0.25\pm \slackBad{}, x_{vw}\in 0.5\pm \slackBad{}$ (recall that $y_{uv} = 1-x_{uv}$, therefore $y_{uv}\in 0.75\pm \slackBad{}, y_{uw}\in 0.75\pm \slackBad{}, y_{vw}\in 0.5\pm \slackBad{}$).
We say that such a bad triangle is centered at $u$.
All other triangles are good triangles.
We say that a chargeable triangle is either a degenerate triangle $uv$ with $x_{uv}\in 0.25\pm \slackBad{}$ ($y_{uv}\in 0.75\pm \slackBad{}$), or a triangle $uvw$ with $x_{uv}\in 0.25\pm \slackBad{}, x_{uw}\in 0.25\pm \slackBad{}, x_{vw}\not\in 0.5\pm \slackBad{}$ ($y_{uv}\in 0.75\pm \slackBad{}, y_{uw}\in 0.75\pm \slackBad{}, y_{vw}\not\in 0.5\pm \slackBad{}$).
A chargeable triangle, by definition, is not a bad triangle.

We also define the function $h:\binom{V}{3}\cup \binom{V}{2} \rightarrow \mathbb{R}$ such that $h(t) = -\charge{}$ if $t$ is a bad triangle, $h(t) = 9 \cdot \charge{}$ if $t$ is chargeable, and $h(t)=0$ otherwise.

We want to show that $\sum_{t\in \binom{V}{3}\cup\binom{V}{2}} h(t) \ge 0$. 
This suffices, because it allows us to overestimate the actual cost of the algorithm by adding $\sum_{t\in \binom{V}{3}\cup\binom{V}{2}} h(t)$.
Repeating a triangle based analysis, now the bad triangles have smaller cost (because $h(t) < 0$ when $t$ is a bad triangle) and therefore their ratio is below $\frac{16}{7}$.
This comes with increasing the ratio of chargeable triangles; however, $h(t)$ is small enough to guarantee that their ratio stays below $\frac{16}{7}$.

\paragraph{Global charging scheme}
Given any $u\in V$, let $G_u = (V_u,E_u)$ be a simple graph such that $V_u = \set{v\in V \mid y_{uv} \in 0.75\pm \slackBad{}}$, and $E_u = \set{vw \in \binom{V_u}{2} \mid y_{uv} \in 0.5 \pm \slackBad{}}$.
The definition of $G_u$ is such that it captures the bad triangles and the chargeable triangles centered at $u$.

\begin{lemma}
For any $u$, the graph $G_u$ is $K_5$-free.
\end{lemma}
\begin{proof}
Assume there exist nodes $v_1, v_2, v_3, v_4, v_5 \in V_u$ that are pairwise connected in $G_u$.
We want to show that this implies $y_{uv_1} > 0.75 + \slackBad{}$, which is a contradiction.

We have $y_{uv_1} \ge y_{uv_1|v_2} + y_{uv_1v_2|v_3} + y_{uv_1v_2v_3|v_4} + y_{uv_1v_2v_3v_4|v_5}$.
We lower bound each of these terms.

\begin{claim} \label{clm:prw}
For any $i,j\in [5], i\ne j$ we have $0.25 - 2\slackBad{} \le y_{uv_i|v_j} \le 0.25 + \slackBad{}$.
\end{claim}
\begin{proof}
For the lower bound: $y_{uv_i|v_j} = y_{uv_i} - y_{uv_iv_j} \ge y_{uv_i} - y_{v_iv_j} \ge (0.75-\slackBad{}) - (0.5+\slackBad{})$.
For the upper bound we have $y_{uv_i|v_j} \le 1 - y_{uv_j} \le 0.25 + \slackBad{}$.
\end{proof}
\begin{claim}
For any $i\in \set{3,4,5}$ we have $y_{uv_1\ldots v_{i-1}|v_i} \ge 0.25 - (3i-2)\slackBad{}$.
\end{claim}
\begin{proof}
We have $y_{uv_1\ldots v_{i-1}|v_i} \ge y_{u|v_i} - \sum_{j=1}^{i-1}(y_{u|v_iv_j} + y_{u|v_i|v_j})$, by union bound as $u|v_iv_j$ and $u|v_i|v_j$ are exactly the events that neither $v_i$ nor $v_j$ is with $u$.

It thus suffices to bound $y_{u|v_iv_j} + y_{u|v_i|v_j}$. Notice that $y_{uv_j|v_i} + y_{u|v_iv_j} + y_{u|v_i|v_j} = y_{u|v_i} = 1 - y_{uv_i} \le 0.25 + \slackBad{}$.
But by \Cref{clm:prw} we have $0.25 - 2\slackBad{} \le y_{uv_j|v_i}$, and therefore $y_{u|v_iv_j} + y_{u|v_i|v_j} \le 3\slackBad{}$.

We conclude that $$y_{uv_1\ldots v_{i-1}|v_i} \ge y_{u|v_i} - \sum_{j=1}^{i-1}(y_{u|v_iv_j} + y_{u|v_i|v_j}) \ge y_{u|v_i} - 3(i-1)\slackBad{} \ge 0.25 - (3i-2)\slackBad{}.$$
\end{proof}

The two claims prove our lemma, as they imply $$y_{uv_1} \ge y_{uv_1|v_2} + y_{uv_1v_2|v_3} + y_{uv_1v_2v_3|v_4} + y_{uv_1v_2v_3v_4|v_5} \ge 1 - 32 \slackBad{} > 0.75 + \slackBad{}.$$
\end{proof}

We are now ready to prove that $\sum_{t\in \binom{V}{3}\cup\binom{V}{2}} h(t) \ge 0$.
\begin{lemma} \label{lem:hPositive}
$\sum_{t\in \binom{V}{3}\cup\binom{V}{2}} h(t) \ge 0$.
\end{lemma}
\begin{proof}
Notice that there is a natural bijection between bad triangles centered at $u$ and edges in $G_u$, as well as a natural bijection between non-degenerate chargeable triangles containing $u$ and non-edges in $G_u$.
Finally, there is a bijection between degenerate chargeable triangles containing $u$ and vertices of $G_u$.

Now fix any $u$. 
As $G_u$ is $K_5$-free, by Turan's theorem we have $|E_u| \le 0.75 \frac{|V_u|^2}{2} \le 0.75\binom{|V_u|}{2} + 0.75 \frac{|V_u|}{2}$.
This means that the ratio $\frac{|E_u|}{|\binom{V_u}{2} \setminus E_u| + |V_u|}$ is at most 3.
This is the ratio of bad triangles centered at $u$ over the chargeable triangles containing $u$.

Now notice that $\sum_u |E_u|$ is exactly the number of bad triangles, while $\sum_u |\binom{V_u}{2}\setminus E_u|+ |V_u|$ counts chargeable triangles and may at most triple count a chargeable triangle.
Therefore, the ratio of bad triangles over chargeable triangles is at most $3\cdot 3 = 9$.

We conclude that $\sum_{t\in \binom{V}{3}\cup\binom{V}{2}} h(t) \ge 0$, as the $h$-value of a chargeable triangle is $9$ times larger than the $h$-value of a bad triangle.
\end{proof}

We now overestimate the actual cost of the algorithm, by adding $\sum_{t\in \binom{V}{3}\cup\binom{V}{2}} h(t) \ge 0$.
Therefore, if we repeat the Triangle-Based analysis, the cost of a bad triangle $uvw$ is decreased by $-\charge{}$, while the cost of chargeable triangles is increased by $9\charge{}$.
The cost of all other triangles remains the same.

\mainTheorem*
\begin{proof}
By \Cref{lem:triangle}, the approximation factor is upper bounded by 
$$\frac{\sum_{uvw\in\binom{V}{3}} \COST{u}{v}[w] + \sum_{uv\in\binom{V}{2}} \COST{u}{v}}{\sum_{uvw\in\binom{V}{3}} \LP{u}{v}[w] + \sum_{uv\in\binom{V}{2}} \LP{u}{v}}.$$
By \Cref{lem:hPositive}, this is upper bounded by 
$$\frac{\sum_{uvw\in\binom{V}{3}} (\COST{u}{v}[w]+h(uvw)) + \sum_{uv\in\binom{V}{2}} (\COST{u}{v}+h(uv))}{\sum_{uvw\in\binom{V}{3}} \LP{u}{v}[w] + \sum_{uv\in\binom{V}{2}} \LP{u}{v}}.$$

We partition the triangles into classes based on the multiset of the supernodes their endpoints belong to.
Let $\mathcal{T}$ be the triangles in such a category, then it suffices to upper bound $\frac{\sum_{T\in \mathcal{T}} (\Cost{T}+h(T))}{\sum_{T\in \mathcal{T}} \Lp{T}}$ for every class; the maximum upper bound obtained also upper bounds the approximation factor.

If the multiset of a class has at most two distinct supernodes, then \Cref{lem:edgesSame,lemma:degenerate-bound-1} show an upper bound of $2.1\le \frac{16}{7}-\gamma$.
If it has three distinct supernodes, then \Cref{lem:edgesSame}, along with 
% \Cref{lemma:degenerate-bound-1,lemma:---_triangle,lemma:--+_triangle,lemma:-++_triangle,lemma:+++_triangle}
\Cref{lemma:---_triangle,lemma:--+_triangle,lemma:-++_triangle,lemma:+++_triangle}
prove an upper bound of $\frac{16}{7}-\gamma$.
\end{proof}

\clearpage
\bibliographystyle{alpha}
\bibliography{bib}

@article{pivoting,
  author    = {Nir Ailon and
               Moses Charikar and
               Alantha Newman},
  title     = {Aggregating inconsistent information: Ranking and clustering},
  journal   = {J. {ACM}},
  volume    = {55},
  number    = {5},
  pages     = {23:1--23:27},
  year      = {2008},
  url       = {https://doi.org/10.1145/1411509.1411513},
  doi       = {10.1145/1411509.1411513},
  bibsource = {dblp computer science bibliography, https://dblp.org},
  note      = {Announced in STOC 2005}
}

@article{charikarTree,
  author       = {Nir Ailon and
                  Moses Charikar},
  title        = {Fitting Tree Metrics: Hierarchical Clustering and Phylogeny},
  journal      = {{SIAM} J. Comput.},
  volume       = {40},
  number       = {5},
  pages        = {1275--1291},
  year         = {2011},
  url          = {https://doi.org/10.1137/100806886},
  doi          = {10.1137/100806886},
  note          = {Announced at FOCS'05},
}

@article{Bansal,
  author    = {Nikhil Bansal and
               Avrim Blum and
               Shuchi Chawla},
  title     = {Correlation Clustering},
  journal   = {Mach. Learn.},
  volume    = {56},
  number    = {1-3},
  pages     = {89--113},
  year      = {2004},
  url       = {https://doi.org/10.1023/B:MACH.0000033116.57574.95},
  doi       = {10.1023/B:MACH.0000033116.57574.95},
  bibsource = {dblp computer science bibliography, https://dblp.org}
}

@article{dnaVertexDeletion,
  author       = {Paola Bonizzoni and
                  Gianluca Della Vedova and
                  Riccardo Dondi and
                  Giancarlo Mauri},
  title        = {Fingerprint Clustering with Bounded Number of Missing Values},
  journal      = {Algorithmica},
  volume       = {58},
  number       = {2},
  pages        = {282--303},
  year         = {2010},
  url          = {https://doi.org/10.1007/s00453-008-9265-0},
  doi          = {10.1007/S00453-008-9265-0},
  bibsource    = {dblp computer science bibliography, https://dblp.org}
}

@inproceedings{clusterLP,
  author       = {Nairen Cao and
                  Vincent Cohen{-}Addad and
                  Euiwoong Lee and
                  Shi Li and
                  Alantha Newman and
                  Lukas Vogl},
  editor       = {Bojan Mohar and
                  Igor Shinkar and
                  Ryan O'Donnell},
  title        = {Understanding the Cluster Linear Program for Correlation Clustering},
  booktitle    = {Proceedings of the 56th Annual {ACM} Symposium on Theory of Computing,
                  {STOC} 2024, Vancouver, BC, Canada, June 24-28, 2024},
  pages        = {1605--1616},
  publisher    = {{ACM}},
  year         = {2024},
  url          = {https://doi.org/10.1145/3618260.3649749},
  doi          = {10.1145/3618260.3649749},
}

@inproceedings{commDet,
  author    = {Yudong Chen and
               Sujay Sanghavi and
               Huan Xu},
  editor    = {Peter L. Bartlett and
               Fernando C. N. Pereira and
               Christopher J. C. Burges and
               L{\'{e}}on Bottou and
               Kilian Q. Weinberger},
  title     = {Clustering Sparse Graphs},
  booktitle = {Advances in Neural Information Processing Systems 25: 26th Annual
               Conference on Neural Information Processing Systems 2012. Proceedings
               of a meeting held December 3-6, 2012, Lake Tahoe, Nevada, United States},
  pages     = {2213--2221},
  year      = {2012},
  url       = {https://proceedings.neurips.cc/paper/2012/hash/1e6e0a04d20f50967c64dac2d639a577-Abstract.html},
  bibsource = {dblp computer science bibliography, https://dblp.org}
}

@article{cutRadius,
  author    = {Moses Charikar and
               Venkatesan Guruswami and
               Anthony Wirth},
  title     = {Clustering with qualitative information},
  journal   = {J. Comput. Syst. Sci.},
  volume    = {71},
  number    = {3},
  pages     = {360--383},
  year      = {2005},
  url       = {https://doi.org/10.1016/j.jcss.2004.10.012},
  doi       = {10.1016/j.jcss.2004.10.012},
  bibsource = {dblp computer science bibliography, https://dblp.org},
  note      = {Announced in FOCS 2003}
}

@inproceedings{sublinearMikkelSTOC,
  author       = {Nairen Cao and
                  Vincent Cohen{-}Addad and
                  Shi Li and
                  Euiwoong Lee and
                  David Rasmussen Lolck and
                  Alantha Newman and
                  Mikkel Thorup and
                  Lukas Vogl and
                  Shuyi Yan and
                  Hanwen Zhang},
  title        = {Solving the Correlation Cluster {LP} in Nearly Linear Time},
  booktitle    = {Proceedings of the 56th Annual {ACM} Symposium on Theory of Computing,
                  {STOC} 2025, Prague, Czech Republic, June 23-27, 2025},
  publisher    = {{ACM}},
  year         = {2025},
  note          = {Accepted at STOC 2025 (not published yet)}
}

@inproceedings{dynNik,
  author       = {Vincent Cohen{-}Addad and
                  Silvio Lattanzi and
                  Andreas Maggiori and
                  Nikos Parotsidis},
  title        = {Dynamic Correlation Clustering in Sublinear Update Time},
  booktitle    = {Forty-first International Conference on Machine Learning, {ICML} 2024,
                  Vienna, Austria, July 21-27, 2024},
  publisher    = {OpenReview.net},
  year         = {2024},
  url          = {https://openreview.net/forum?id=3YG55Lbcnr},
  bibsource    = {dblp computer science bibliography, https://dblp.org}
}

@inproceedings{parallel,
  author    = {Vincent Cohen{-}Addad and
               Silvio Lattanzi and
               Slobodan Mitrovic and
               Ashkan Norouzi{-}Fard and
               Nikos Parotsidis and
               Jakub Tarnawski},
  editor    = {Marina Meila and
               Tong Zhang},
  title     = {Correlation Clustering in Constant Many Parallel Rounds},
  booktitle = {Proceedings of the 38th International Conference on Machine Learning,
               {ICML} 2021, 18-24 July 2021, Virtual Event},
  series    = {Proceedings of Machine Learning Research},
  volume    = {139},
  pages     = {2069--2078},
  publisher = {{PMLR}},
  year      = {2021},
  url       = {http://proceedings.mlr.press/v139/cohen-addad21b.html},
  bibsource = {dblp computer science bibliography, https://dblp.org}
}

@inproceedings{sub2,
  author       = {Vincent Cohen{-}Addad and
                  Euiwoong Lee and
                  Alantha Newman},
  title        = {Correlation Clustering with {Sherali-Adams}},
  booktitle    = {63rd {IEEE} Annual Symposium on Foundations of Computer Science, {FOCS}
                  2022, Denver, CO, USA, October 31 - November 3, 2022},
  pages        = {651--661},
  publisher    = {{IEEE}},
  year         = {2022},
  url          = {https://doi.org/10.1109/FOCS54457.2022.00068},
  doi          = {10.1109/FOCS54457.2022.00068},
  bibsource    = {dblp computer science bibliography, https://dblp.org}
}

@inproceedings{apx173,
  author       = {Vincent Cohen{-}Addad and
                  Euiwoong Lee and
                  Shi Li and
                  Alantha Newman},
  title        = {Handling Correlated Rounding Error via Preclustering: {A} 1.73-approximation
                  for Correlation Clustering},
  booktitle    = {64th {IEEE} Annual Symposium on Foundations of Computer Science, {FOCS}
                  2023, Santa Cruz, CA, USA, November 6-9, 2023},
  pages        = {1082--1104},
  publisher    = {{IEEE}},
  year         = {2023},
  url          = {https://doi.org/10.1109/FOCS57990.2023.00065},
  doi          = {10.1109/FOCS57990.2023.00065},
  bibsource    = {dblp computer science bibliography, https://dblp.org}
}

@inproceedings{combinatorialMikkel,
  author       = {Vincent Cohen{-}Addad and
                  David Rasmussen Lolck and
                  Marcin Pilipczuk and
                  Mikkel Thorup and
                  Shuyi Yan and
                  Hanwen Zhang},
  editor       = {Bojan Mohar and
                  Igor Shinkar and
                  Ryan O'Donnell},
  title        = {Combinatorial Correlation Clustering},
  booktitle    = {Proceedings of the 56th Annual {ACM} Symposium on Theory of Computing,
                  {STOC} 2024, Vancouver, BC, Canada, June 24-28, 2024},
  pages        = {1617--1628},
  publisher    = {{ACM}},
  year         = {2024},
  url          = {https://doi.org/10.1145/3618260.3649712},
  doi          = {10.1145/3618260.3649712},
}

@article{barrierPseudometric,
  author       = {Dahoon Lee and
                  Chenglin Fan and
                  Euiwoong Lee},
  title        = {Improved Approximation Algorithms for Chromatic and Pseudometric-Weighted
                  Correlation Clustering},
  journal      = {CoRR},
  volume       = {abs/2505.21939},
  year         = {2025},
  url          = {https://doi.org/10.48550/arXiv.2505.21939},
  doi          = {10.48550/ARXIV.2505.21939},
  eprinttype    = {arXiv},
  eprint       = {2505.21939},
}

@article{chromaticClusterLP,
  author       = {Dahoon Lee and
                  Chenglin Fan and
                  Euiwoong Lee},
  title        = {1.64-Approximation for Chromatic Correlation Clustering via Chromatic
                  Cluster {LP}},
  journal      = {CoRR},
  volume       = {abs/2507.15417},
  year         = {2025},
  url          = {https://doi.org/10.48550/arXiv.2507.15417},
  doi          = {10.48550/ARXIV.2507.15417},
  eprinttype    = {arXiv},
  eprint       = {2507.15417},
}

@article{robust,
  author       = {Martin Farach and
                  Sampath Kannan and
                  Tandy J. Warnow},
  title        = {A Robust Model for Finding Optimal Evolutionary Trees},
  journal      = {Algorithmica},
  volume       = {13},
  number       = {1/2},
  pages        = {155--179},
  year         = {1995},
  note         = {Announced at STOC 1993},
}

@article{EvangelosSTACS,
  author       = {Nick Fischer and
                  Jonas Klausen and
                  Evangelos Kipouridis and
                  Mikkel Thorup},
  title        = {A Faster Algorithm for Constrained Correlation Clustering},
  journal      = {CoRR},
  volume       = {abs/2501.03154},
  year         = {2024},
  url          = {https://doi.org/10.48550/arXiv.2501.03154},
  doi          = {10.48550/ARXIV.2501.03154},
  eprinttype    = {arXiv},
  eprint       = {2501.03154},
  note          = {To appear at STACS 2025}
}

@misc{cao2026clusterdeletionhardapproximate,
      title={Cluster Deletion is as Hard to Approximate as Vertex Cover}, 
      author={Yixin Cao and Ying Xu},
      year={2026},
      eprint={2608.04883},
      archivePrefix={arXiv},
      primaryClass={cs.DS},
      url={https://arxiv.org/abs/2608.04883}, 
}

@article{clusterDeletionSpecialClassesBoundedClique,
  author       = {Nicola Galesi and
                  Tony Huynh and
                  Fariba Ranjbar},
  title        = {Cluster deletion and clique partitioning in graphs with bounded clique
                  number},
  journal      = {CoRR},
  volume       = {abs/2505.00922},
  year         = {2025},
  url          = {https://doi.org/10.48550/arXiv.2505.00922},
  doi          = {10.48550/ARXIV.2505.00922},
  eprinttype    = {arXiv},
  eprint       = {2505.00922},
  bibsource    = {dblp computer science bibliography, https://dblp.org}
}

@article{local,
  author    = {Gregory J. Puleo and
               Olgica Milenkovic},
  title     = {Correlation Clustering with Constrained Cluster Sizes and Extended
               Weights Bounds},
  journal   = {{SIAM} J. Optim.},
  volume    = {25},
  number    = {3},
  pages     = {1857--1872},
  year      = {2015},
  url       = {https://doi.org/10.1137/140994198},
  doi       = {10.1137/140994198},
  bibsource = {dblp computer science bibliography, https://dblp.org}
}

@article{deterministicPivoting,
  author    = {Anke van Zuylen and
               David P. Williamson},
  title     = {Deterministic Pivoting Algorithms for Constrained Ranking and Clustering
               Problems},
  journal   = {Math. Oper. Res.},
  volume    = {34},
  number    = {3},
  pages     = {594--620},
  year      = {2009},
  url       = {https://doi.org/10.1287/moor.1090.0385},
  doi       = {10.1287/moor.1090.0385},
  bibsource = {dblp computer science bibliography, https://dblp.org},
  note      = {Announced in SODA 2007},
}

@inproceedings{differentialPrivacy,
  author    = {Mark Bun and
               Marek Eli{\'{a}}s and
               Janardhan Kulkarni},
  editor    = {Marina Meila and
               Tong Zhang},
  title     = {Differentially Private Correlation Clustering},
  booktitle = {Proceedings of the 38th International Conference on Machine Learning,
               {ICML} 2021, 18-24 July 2021, Virtual Event},
  series    = {Proceedings of Machine Learning Research},
  volume    = {139},
  pages     = {1136--1146},
  publisher = {{PMLR}},
  year      = {2021},
  url       = {http://proceedings.mlr.press/v139/bun21a.html},
  bibsource = {dblp computer science bibliography, https://dblp.org}
}

@inproceedings{deterministic,
  author    = {Nate Veldt},
  editor    = {Kamalika Chaudhuri and
               Stefanie Jegelka and
               Le Song and
               Csaba Szepesv{\'{a}}ri and
               Gang Niu and
               Sivan Sabato},
  title     = {Correlation Clustering via Strong Triadic Closure Labeling: Fast Approximation
               Algorithms and Practical Lower Bounds},
  booktitle = {International Conference on Machine Learning, {ICML} 2022, 17-23 July
               2022, Baltimore, Maryland, {USA}},
  series    = {Proceedings of Machine Learning Research},
  volume    = {162},
  pages     = {22060--22083},
  publisher = {{PMLR}},
  year      = {2022},
  url       = {https://proceedings.mlr.press/v162/veldt22a.html},
  bibsource = {dblp computer science bibliography, https://dblp.org}
}

@inproceedings{veldtFaster2,
  author       = {Vicente Balmaseda and
                  Ying Xu and
                  Yixin Cao and
                  Nate Veldt},
  title        = {Combinatorial Approximations for Cluster Deletion: Simpler, Faster,
                  and Better},
  booktitle    = {Forty-first International Conference on Machine Learning, {ICML} 2024,
                  Vienna, Austria, July 21-27, 2024},
  publisher    = {OpenReview.net},
  year         = {2024},
  url          = {https://openreview.net/forum?id=FpbKoIPHxb},
  bibsource    = {dblp computer science bibliography, https://dblp.org}
}

@inproceedings{image2,
  author    = {Julian Yarkony and
               Alexander T. Ihler and
               Charless C. Fowlkes},
  editor    = {Andrew W. Fitzgibbon and
               Svetlana Lazebnik and
               Pietro Perona and
               Yoichi Sato and
               Cordelia Schmid},
  title     = {Fast Planar Correlation Clustering for Image Segmentation},
  booktitle = {Computer Vision - {ECCV} 2012 - 12th European Conference on Computer
               Vision, Florence, Italy, October 7-13, 2012, Proceedings, Part {VI}},
  series    = {Lecture Notes in Computer Science},
  volume    = {7577},
  pages     = {568--581},
  publisher = {Springer},
  year      = {2012},
  url       = {https://doi.org/10.1007/978-3-642-33783-3\_41},
  doi       = {10.1007/978-3-642-33783-3\_41},
  bibsource = {dblp computer science bibliography, https://dblp.org}
}

@inproceedings{autoLabelAgrawal,
  author    = {Rakesh Agrawal and
               Alan Halverson and
               Krishnaram Kenthapadi and
               Nina Mishra and
               Panayiotis Tsaparas},
  editor    = {Ricardo Baeza{-}Yates and
               Paolo Boldi and
               Berthier A. Ribeiro{-}Neto and
               Berkant Barla Cambazoglu},
  title     = {Generating labels from clicks},
  booktitle = {Proceedings of the Second International Conference on Web Search and
               Web Data Mining, {WSDM} 2009, Barcelona, Spain, February 9-11, 2009},
  pages     = {172--181},
  publisher = {{ACM}},
  year      = {2009},
  url       = {https://doi.org/10.1145/1498759.1498824},
  doi       = {10.1145/1498759.1498824},
  bibsource = {dblp computer science bibliography, https://dblp.org}
}

@inproceedings{autoLabelChakrabarti,
  author    = {Deepayan Chakrabarti and
               Ravi Kumar and
               Kunal Punera},
  editor    = {Jinpeng Huai and
               Robin Chen and
               Hsiao{-}Wuen Hon and
               Yunhao Liu and
               Wei{-}Ying Ma and
               Andrew Tomkins and
               Xiaodong Zhang},
  title     = {A graph-theoretic approach to webpage segmentation},
  booktitle = {Proceedings of the 17th International Conference on World Wide Web,
               {WWW} 2008, Beijing, China, April 21-25, 2008},
  pages     = {377--386},
  publisher = {{ACM}},
  year      = {2008},
  url       = {https://doi.org/10.1145/1367497.1367549},
  doi       = {10.1145/1367497.1367549},
  bibsource = {dblp computer science bibliography, https://dblp.org}
}

@inproceedings{image1,
  author    = {Sungwoong Kim and
               Sebastian Nowozin and
               Pushmeet Kohli and
               Chang Dong Yoo},
  editor    = {John Shawe{-}Taylor and
               Richard S. Zemel and
               Peter L. Bartlett and
               Fernando C. N. Pereira and
               Kilian Q. Weinberger},
  title     = {Higher-Order Correlation Clustering for Image Segmentation},
  booktitle = {Advances in Neural Information Processing Systems 24: 25th Annual
               Conference on Neural Information Processing Systems 2011. Proceedings
               of a meeting held 12-14 December 2011, Granada, Spain},
  pages     = {1530--1538},
  year      = {2011},
  url       = {https://proceedings.neurips.cc/paper/2011/hash/98d6f58ab0dafbb86b083a001561bb34-Abstract.html},
  bibsource = {dblp computer science bibliography, https://dblp.org}
}

@inproceedings{duplicate,
  author    = {Arvind Arasu and
               Christopher R{\'{e}} and
               Dan Suciu},
  editor    = {Yannis E. Ioannidis and
               Dik Lun Lee and
               Raymond T. Ng},
  title     = {Large-Scale Deduplication with Constraints Using Dedupalog},
  booktitle = {Proceedings of the 25th International Conference on Data Engineering,
               {ICDE} 2009, March 29 2009 - April 2 2009, Shanghai, China},
  pages     = {952--963},
  publisher = {{IEEE} Computer Society},
  year      = {2009},
  url       = {https://doi.org/10.1109/ICDE.2009.43},
  doi       = {10.1109/ICDE.2009.43},
  bibsource = {dblp computer science bibliography, https://dblp.org}
}

@inproceedings{NearOptimal2,
  author    = {Shuchi Chawla and
               Konstantin Makarychev and
               Tselil Schramm and
               Grigory Yaroslavtsev},
  editor    = {Rocco A. Servedio and
               Ronitt Rubinfeld},
  title     = {Near Optimal {LP} Rounding Algorithm for Correlation Clustering on
               Complete and Complete $k$-partite Graphs},
  booktitle = {Proceedings of the Forty-Seventh Annual {ACM} on Symposium on Theory
               of Computing, {STOC} 2015, Portland, OR, USA, June 14-17, 2015},
  pages     = {219--228},
  publisher = {{ACM}},
  year      = {2015},
  url       = {https://doi.org/10.1145/2746539.2746604},
  doi       = {10.1145/2746539.2746604},
  bibsource = {dblp computer science bibliography, https://dblp.org}
}

@article{clusteringEnsembles,
  author    = {Francesco Bonchi and
               Aristides Gionis and
               Antti Ukkonen},
  title     = {Overlapping correlation clustering},
  journal   = {Knowl. Inf. Syst.},
  volume    = {35},
  number    = {1},
  pages     = {1--32},
  year      = {2013},
  url       = {https://doi.org/10.1007/s10115-012-0522-9},
  doi       = {10.1007/s10115-012-0522-9},
  bibsource = {dblp computer science bibliography, https://dblp.org}
}

@inproceedings{sublinear,
  author       = {Sepehr Assadi and
                  Chen Wang},
  editor       = {Mark Braverman},
  title        = {Sublinear Time and Space Algorithms for Correlation Clustering via
                  Sparse-Dense Decompositions},
  booktitle    = {13th Innovations in Theoretical Computer Science Conference, {ITCS}
                  2022, January 31 - February 3, 2022, Berkeley, CA, {USA}},
  series       = {LIPIcs},
  volume       = {215},
  pages        = {10:1--10:20},
  publisher    = {Schloss Dagstuhl - Leibniz-Zentrum f{\"{u}}r Informatik},
  year         = {2022},
  url          = {https://doi.org/10.4230/LIPIcs.ITCS.2022.10},
  doi          = {10.4230/LIPICS.ITCS.2022.10},
  bibsource    = {dblp computer science bibliography, https://dblp.org}
}

@inproceedings{BehMany,
  author       = {Soheil Behnezhad and
                  Moses Charikar and
                  Weiyun Ma and
                  Li{-}Yang Tan},
  title        = {Almost 3-Approximate Correlation Clustering in Constant Rounds},
  booktitle    = {63rd {IEEE} Annual Symposium on Foundations of Computer Science, {FOCS}
                  2022, Denver, CO, USA, October 31 - November 3, 2022},
  pages        = {720--731},
  publisher    = {{IEEE}},
  year         = {2022},
  url          = {https://doi.org/10.1109/FOCS54457.2022.00074},
  doi          = {10.1109/FOCS54457.2022.00074},
}

@inproceedings{BehSingle,
  author       = {Soheil Behnezhad and
                  Moses Charikar and
                  Weiyun Ma and
                  Li{-}Yang Tan},
  editor       = {Nikhil Bansal and
                  Viswanath Nagarajan},
  title        = {Single-Pass Streaming Algorithms for Correlation Clustering},
  booktitle    = {Proceedings of the 2023 {ACM-SIAM} Symposium on Discrete Algorithms,
                  {SODA} 2023, Florence, Italy, January 22-25, 2023},
  pages        = {819--849},
  publisher    = {{SIAM}},
  year         = {2023},
  url          = {https://doi.org/10.1137/1.9781611977554.ch33},
  doi          = {10.1137/1.9781611977554.CH33},
}

@inproceedings{dynamicStream,
  author       = {M{\'{e}}lanie Cambus and
                  Fabian Kuhn and
                  Etna Lindy and
                  Shreyas Pai and
                  Jara Uitto},
  editor       = {David P. Woodruff},
  title        = {A {(3} + {\(\varepsilon\)})-Approximate Correlation Clustering Algorithm
                  in Dynamic Streams},
  booktitle    = {Proceedings of the 2024 {ACM-SIAM} Symposium on Discrete Algorithms,
                  {SODA} 2024, Alexandria, VA, USA, January 7-10, 2024},
  pages        = {2861--2880},
  publisher    = {{SIAM}},
  year         = {2024},
  url          = {https://doi.org/10.1137/1.9781611977912.101},
  doi          = {10.1137/1.9781611977912.101},
  bibsource    = {dblp computer science bibliography, https://dblp.org}
}

@inproceedings{sub3parallel,
  author       = {Nairen Cao and
                  Shang{-}En Huang and
                  Hsin{-}Hao Su},
  editor       = {David P. Woodruff},
  title        = {Breaking 3-Factor Approximation for Correlation Clustering in Polylogarithmic
                  Rounds},
  booktitle    = {Proceedings of the 2024 {ACM-SIAM} Symposium on Discrete Algorithms,
                  {SODA} 2024, Alexandria, VA, USA, January 7-10, 2024},
  pages        = {4124--4154},
  publisher    = {{SIAM}},
  year         = {2024},
  url          = {https://doi.org/10.1137/1.9781611977912.143},
  doi          = {10.1137/1.9781611977912.143},
  bibsource    = {dblp computer science bibliography, https://dblp.org}
}

@inproceedings{MakarySingle,
  author       = {Konstantin Makarychev and
                  Sayak Chakrabarty},
  editor       = {Alice Oh and
                  Tristan Naumann and
                  Amir Globerson and
                  Kate Saenko and
                  Moritz Hardt and
                  Sergey Levine},
  title        = {Single-Pass Pivot Algorithm for Correlation Clustering. {K}eep it simple!},
  booktitle    = {Advances in Neural Information Processing Systems 36: Annual Conference
                  on Neural Information Processing Systems 2023, NeurIPS 2023, New Orleans,
                  LA, USA, December 10 - 16, 2023},
  year         = {2023},
  url          = {http://papers.nips.cc/paper\_files/paper/2023/hash/149ad6e32c08b73a3ecc3d11977fcc47-Abstract-Conference.html},
  bibsource    = {dblp computer science bibliography, https://dblp.org}
}

@article{vc-hard-to-approx,
  author       = {Subhash Khot and
                  Oded Regev},
  title        = {Vertex cover might be hard to approximate to within $2 - \varepsilon$},
  journal      = {J. Comput. Syst. Sci.},
  volume       = {74},
  number       = {3},
  pages        = {335--349},
  year         = {2008},
  url          = {https://doi.org/10.1016/j.jcss.2007.06.019},
  doi          = {10.1016/J.JCSS.2007.06.019}
}

@article{fpt,
  author    = {Fedor V. Fomin and
               Stefan Kratsch and
               Marcin Pilipczuk and
               Michal Pilipczuk and
               Yngve Villanger},
  title     = {Tight bounds for parameterized complexity of Cluster Editing with
               a small number of clusters},
  journal   = {J. Comput. Syst. Sci.},
  volume    = {80},
  number    = {7},
  pages     = {1430--1447},
  year      = {2014},
  url       = {https://doi.org/10.1016/j.jcss.2014.04.015},
  doi       = {10.1016/j.jcss.2014.04.015},
  bibsource = {dblp computer science bibliography, https://dblp.org}
}

@inproceedings{ugc,
  author       = {Subhash Khot},
  editor       = {John H. Reif},
  title        = {On the power of unique 2-prover 1-round games},
  booktitle    = {Proceedings on 34th Annual {ACM} Symposium on Theory of Computing,
                  May 19-21, 2002, Montr{\'{e}}al, Qu{\'{e}}bec, Canada},
  pages        = {767--775},
  publisher    = {{ACM}},
  year         = {2002},
  url          = {https://doi.org/10.1145/509907.510017},
  doi          = {10.1145/509907.510017}
}

@InProceedings{kalavas,
  title = 	 {Towards Better-than-2 Approximation for Constrained Correlation Clustering},
  author =       {Kalavas, Andreas and Kipouridis, Evangelos and Varma, Nithin},
  booktitle = 	 {Proceedings of the 42nd International Conference on Machine Learning},
  pages = 	 {28721--28734},
  year = 	 {2025},
  editor = 	 {Singh, Aarti and Fazel, Maryam and Hsu, Daniel and Lacoste-Julien, Simon and Berkenkamp, Felix and Maharaj, Tegan and Wagstaff, Kiri and Zhu, Jerry},
  volume = 	 {267},
  series = 	 {Proceedings of Machine Learning Research},
  month = 	 {13--19 Jul},
  publisher =    {PMLR},
  url = 	 {https://proceedings.mlr.press/v267/kalavas25a.html}
}

@article{disambiguation,
  author    = {Dmitri V. Kalashnikov and
               Zhaoqi Chen and
               Sharad Mehrotra and
               Rabia Nuray{-}Turan},
  title     = {Web People Search via Connection Analysis},
  journal   = {{IEEE} Trans. Knowl. Data Eng.},
  volume    = {20},
  number    = {11},
  pages     = {1550--1565},
  year      = {2008},
  url       = {https://doi.org/10.1109/TKDE.2008.78},
  doi       = {10.1109/TKDE.2008.78},
  bibsource = {dblp computer science bibliography, https://dblp.org}
}

@article{Ma3Infinity,
  author       = {Bin Ma and
                  Lusheng Wang and
                  Louxin Zhang},
  title        = {Fitting Distances by Tree Metrics with Increment Error},
  journal      = {J. Comb. Optim.},
  volume       = {3},
  number       = {2-3},
  pages        = {213--225},
  year         = {1999},
  url          = {https://doi.org/10.1023/A:1009837726913},
  doi          = {10.1023/A:1009837726913},
  bibsource    = {dblp computer science bibliography, https://dblp.org}
}

@inproceedings{communityVeldt,
  author       = {Nate Veldt and
                  David F. Gleich and
                  Anthony Wirth},
  editor       = {Pierre{-}Antoine Champin and
                  Fabien Gandon and
                  Mounia Lalmas and
                  Panagiotis G. Ipeirotis},
  title        = {A Correlation Clustering Framework for Community Detection},
  booktitle    = {Proceedings of the 2018 World Wide Web Conference on World Wide Web,
                  {WWW} 2018, Lyon, France, April 23-27, 2018},
  pages        = {439--448},
  publisher    = {{ACM}},
  year         = {2018},
  url          = {https://doi.org/10.1145/3178876.3186110},
  doi          = {10.1145/3178876.3186110},
  bibsource    = {dblp computer science bibliography, https://dblp.org}
}

@inproceedings{CrosslingualGaelZ07,
  author       = {Jurgen Van Gael and
                  Xiaojin Zhu},
  editor       = {Manuela M. Veloso},
  title        = {Correlation Clustering for Crosslingual Link Detection},
  booktitle    = {{IJCAI} 2007, Proceedings of the 20th International Joint Conference
                  on Artificial Intelligence, Hyderabad, India, January 6-12, 2007},
  pages        = {1744--1749},
  year         = {2007},
  url          = {http://ijcai.org/Proceedings/07/Papers/282.pdf},
  bibsource    = {dblp computer science bibliography, https://dblp.org}
}

@inproceedings{VeldtConstrained,
  author       = {Nate Veldt},
  title        = {A Simple and Fast (3+$\varepsilon$)-approximation for Constrained Correlation Clustering},
  booktitle    = {9th Symposium on Simplicity in Algorithms, SOSA@SODA 2026, Vancouver, Canada, January 12-13, 2026},
  publisher    = {{SIAM}},
  year         = {2026}
}

@inproceedings{l0TreeLower,
  author       = {Debarati Das and Evangelos Kipouridis and Joachim Spoerhase},
  title        = {Tree Violation Distance under Constraints},
  booktitle    = {9th Symposium on Simplicity in Algorithms, SOSA@SODA 2026, Vancouver, Canada, January 12-13, 2026},
  publisher    = {{SIAM}},
  year         = {2026}
}

@inproceedings{cons1,
  title={Constrained $k$-means clustering with background knowledge},
  author={Wagstaff, Kiri and Cardie, Claire and Rogers, Seth and Schr{\"o}dl, Stefan and others},
  booktitle={ICML},
  volume={1},
  pages={577--584},
  year={2001}
}

@inproceedings{Cons3,
  title={Flexible constrained spectral clustering},
  author={Wang, Xiang and Davidson, Ian},
  booktitle={Proceedings of the 16th ACM SIGKDD international conference on Knowledge discovery and data mining},
  pages={563--572},
  year={2010}
}

@inproceedings{cons4,
  title={Fuzzy clustering with spatial constraints},
  author={Pham, Dzung L},
  booktitle={Proceedings. International Conference on Image Processing},
  volume={2},
  pages={II--II},
  year={2002},
  organization={IEEE}
}

\clearpage
\appendix
\section{\texorpdfstring{Properties of $f^+,f^-$}{Properties of f+, f-}} \label{app:proofsForFunctions}
In this section we present several properties of $f^+$ and $f^-$:

$$
f^+(x) = \begin{cases}
	\sqrt{1 - \frac{20}{9} x} & \text{if } x < 0.25 \\
	\frac{13}{12} - \frac{5}{3} x & \text{if } x < 0.5,
	\qquad
	f^-(x) = \begin{cases}
		1 - x & \text{if } x < 0.5 \\
		0 & \text{if } x \geq 0.5
	\end{cases} \\
	0 & \text{if } x \geq 0.5
\end{cases}
$$

\derCont*
\begin{proof}
	The only potential junction within $[0,0.5)$ is at $x=0.25$. We verify the values and one-sided derivatives match.
	For the values,
	\begin{align*}
		f^+(0.25) &= \frac{13}{12} - \frac{5}{3}\cdot 0.25 \;=\; \frac{2}{3},\\
		f^+(0.25^-) &= \lim_{x\to 0.25^-} \sqrt{1 - \frac{20}{9}x}
		= \sqrt{1 - \frac{20}{9}\cdot 0.25} \;=\; \frac{2}{3}.
	\end{align*}
	For the derivatives, for $x\in[0,0.25)$ we have
	$$
	f'^+(x) \;=\; -\frac{10}{9}\cdot\frac{1}{\sqrt{\,1-\frac{20}{9}x\,}},
	$$
	and for $x\in[0.25,0.5)$ we have $f'^+(x)=-\frac{5}{3}$. Hence
	$$
	f'^+(0.25^-) \;=\; \lim_{x\to 0.25^-} \frac{-10}{9\sqrt{1-\frac{20}{9}x}}
	\;=\; \frac{-10}{9\cdot \frac{2}{3}} \;=\; -\frac{5}{3} \;=\; f'^+(0.25).
	$$
	Thus $f^+$ and its first derivative are continuous on $[0,0.5)$.
\end{proof}

\funcIneq*
\begin{proof}
	$(i)$ By definition of $f^+(x)$, it is either equal to $\sqrt{1-\frac{20}{9} x}$ or $\frac{13}{12} - \frac{5}{3}x$, therefore it suffices to prove for any $x\in [0, 0.5)$
	$$
	\sqrt{1-\frac{20}{9} x} \leq \frac{13}{12} - \frac{5}{3}x.
	$$
	Both sides are nonnegative, therefore it suffices to prove :
	$$
	1 - \frac{20}{9}x \;\le\; \Bigl(\frac{13}{12} - \frac{5}{3}x\Bigr)^2.
	$$
	Expanding gives
	$$
	\Bigl(\frac{13}{12} - \frac{5}{3}x\Bigr)^2 - \Bigl(1 - \frac{20}{9}x\Bigr)
	= \frac{25}{9}\Bigl(x - \frac{1}{4}\Bigr)^2 \;\ge\; 0,
	$$
	as required.
	
	\noindent$(ii)$ The branches $\sqrt{1 - \frac{20}{9}x}$ and $\frac{13}{12} - \frac{5}{3}x$ are concave. By Lemma~\ref{lemma:der-cont}, $f^+$ is concave on $[0,0.5)$. Hence for $x\in[0,0.5)$,
	$$
	f^+(x) \;\ge\; 2x\cdot f^+(0.5^-) + (1-2x)\cdot f^+(0)
	\;=\; 2x\cdot\frac{1}{4} + (1-2x)\cdot 1
	\;=\; 1 - \frac{3}{2}x,
	$$
	and the tangent bound at $x=0$ gives
	$$
	f^+(x) \;\le\; f^+(0) + x\,f'^+(0) \;=\; 1 - \frac{10}{9}x.
	$$
\end{proof}

\fProduct*
\begin{proof}
    Note that for $x\in [0, 0.25)$ we have
    $$
    f'^+(x) = -\frac{10}{9} \cdot \frac{1}{\sqrt{1 - \frac{20}{9} x}},
    $$
    which is a decreasing function and is negative on the entire domain. For $x \in [0.25, 0.5)$, by definition of $f^+$, we have $f'^+(x) = \frac{-5}{3}$. By \Cref{lemma:der-cont}, $f'^+$ is continuous. Consequently $f'^+$ is a continuous decreasing function which is negative on the entire domain. We have
    $$
    \Bigl(\ln (f^+(x))\Bigr)' = \frac{f'^+(x)}{f^+(x)}.
    $$
    Since $f'^+(x)$ is decreasing and negative, its absolute value is increasing. By definition of $f^+(x)$, it is decreasing and positive, so its absolute value is also decreasing. Therefore, $\left\vert \frac{f'^+(x)}{f^+(x)}\right\vert$ is increasing, but since $f'^+(x)$ is negative, we conclude that the first derivative of $\ln (f^+(x))$ is decreasing, and hence it is concave.
	\REM{We divide into subcases.
	\emph{$(i)$ Both $x,y\ge 0.25$.} Then
	\begin{align*}
		f^+(x)f^+(y)
		&= \Bigl(\frac{13}{12} - \frac{5}{3}x\Bigr)\Bigl(\frac{13}{12} - \frac{5}{3}y\Bigr)
		= \frac{13^2}{12^2} - \frac{13\cdot 5}{12\cdot 3}(x+y) + \frac{25}{9}xy,\\
		f^+\Bigl(\frac{x+y}{2}\Bigr)^2
		&= \Bigl(\frac{13}{12} - \frac{5}{6}(x+y)\Bigr)^2
		= \frac{13^2}{12^2} - \frac{13\cdot 5}{12\cdot 3}(x+y) + \frac{25}{9}\Bigl(\frac{x+y}{2}\Bigr)^2,
	\end{align*}
	and $(\frac{x+y}{2})^2\ge xy$ proves the claim.
	\emph{$(ii)$ One of $x,y$ is $\ge 0.25$ and the other $<0.25$.} WLOG $x\ge 0.25>y$. By Lemma~\ref{lemma:func-ineq}$(i)$,
	\begin{align*}
		f^+(x)f^+(y)
		&\le \Bigl(\frac{13}{12} - \frac{5}{3}x\Bigr)\Bigl(\frac{13}{12} - \frac{5}{3}y\Bigr) \\
		&= \frac{13^2}{12^2} - \frac{13\cdot 5}{12\cdot 3}(x+y) + \frac{25}{9}xy
		\;\le\;
		\frac{13^2}{12^2} - \frac{13\cdot 5}{12\cdot 3}(x+y) + \frac{25}{9}\Bigl(\frac{x+y}{2}\Bigr)^2 \\
		&= f^+\Bigl(\frac{x+y}{2}\Bigr)^2.
	\end{align*}
	\emph{$(iii)$ Both $x, y \leq 0.25$.} It suffices to prove
	$$
	f^+(x)^2 f^+(y)^2 \leq f^+(\frac{x+y}{2})^4.
	$$
	Expanding gives
	\begin{align*}
		f^+(x)^2 f^+(y)^2 & = (1 - \frac{20}{9} x) (1 - \frac{20}{9} y) \\
		& = 1 - \frac{20}{9} (x + y) + \frac{400}{81} xy \\
		& \leq 1 - \frac{20}{9} (x + y) + \frac{400}{81} (\frac{x+y}{2})^2 = f^+(\frac{x+y}{2})^4.
	\end{align*}
	Combining these results, the proof is complete.}
\end{proof}

\funcInequality*
\begin{proof}
    Note that if $x \in [0.5, 1]$, both $f^+(x)$, and $f^-(x)$ are zero and the inequality follows. If $x \in [0.25, 0.5)$, by definition of $f^+, f^-$, and $x \geq 0.25$, we have
    $$
    f^+(x) = \frac{13}{12} - \frac{5}{3} x = (1 - x) + (\frac{1}{12} - \frac{2}{3} x) \leq (1 - x) + (\frac{1}{12} - \frac{2}{12}) \leq (1 - x) = f^-(x).
    $$
    For the final case $x \in [0, 0.25)$, we must prove $\sqrt{1 - \frac{20}{9} x} \leq 1 - x$. By squarring both sides, this is equivalent to
    $$
    1 - \frac{20}{9} x \leq 1 - 2x + x^2,
    $$
    which simplifies to
    $$
    0 \leq \frac{2}{9} x + x^2,
    $$
    and this holds.
\end{proof}

\end{document}